%% file: main.tex
\documentclass[english,12pt]{article}

\usepackage{amsmath}
\usepackage{times}
\usepackage{bm}
\usepackage{natbib}
\usepackage{fullpage}
\usepackage{comment}
\usepackage[plain,noend]{algorithm2e}
\usepackage{algorithmic}
\usepackage{algorithm2e}
\usepackage{amssymb}
\usepackage{mathrsfs}
\usepackage{placeins}
\usepackage{bbm}
\usepackage{float} 
\usepackage[hidelinks]{hyperref}
\usepackage{xcolor}
\usepackage{amsthm}

\usepackage{cleveref}
\crefname{assumption}{Assumption}{Assumptions}
\crefname{remark}{Remark}{Remarks}
\crefname{proposition}{Proposition}{Propositions}
\crefname{theorem}{Theorem}{Theorems}
\crefname{section}{Section}{Section}
\crefname{lemma}{Lemma}{Lemma}
\crefname{algorithm}{Algorithm}{Algorithms}
\crefname{example}{Example}{Examples}
\crefname{figure}{Figure}{Figure}
\crefname{appendix}{Appendix}{Appendix}
\crefname{equation}{equation}{equation}

\newtheorem{theorem}{Theorem}
\newtheorem{corollary}{Corollary}
\newtheorem{proposition}[theorem]{Proposition} 
\newtheorem{assumption}{Assumption} 
\newtheorem{remark}{Remark} 
\newtheorem{lemma}[theorem]{Lemma}
\include{macros}
\def\aa{\hyp}
\def\aatt{\widehat{\hyp}} 

\usepackage{enumerate}

\allowdisplaybreaks
\begin{document}

\title{Posterior computation with the Gibbs zig-zag sampler}

\author{Matthias Sachs$^{1}$\footnote{The two authors contributed equally to this paper.} $^,$\footnote{Corresponding author.} \\ {m,sachs@bham.ac.uk} 
\and 
Deborshee Sen$^{2,}$\footnotemark[1]\\ {ds2469@bath.ac.uk}  
\and 
Jianfeng Lu$^{3}$ \\
{jianfeng@math.duke.edu}
\and
David Dunson$^{3,4}$ \\ {dunson@duke.edu} }

\date{$^1$School of Mathematics, University of Birmingham \\
$^2$Department of Mathematical Sciences, University of Bath \\
$^3$Department of Statistical Science, Duke University \\
$^4$Department of Mathematics, Duke University \\}

\maketitle 

\begin{abstract}
\noindent 
An intriguing new class of piecewise deterministic Markov processes (PDMPs) has recently been proposed as an alternative to Markov chain Monte Carlo (MCMC). We propose a new class of PDMPs termed Gibbs zig-zag samplers, which allow parameters to be updated in blocks with a zig-zag sampler applied to certain parameters and traditional MCMC-style updates to others. We demonstrate the flexibility of this framework on posterior sampling for logistic models with shrinkage priors for high-dimensional regression and random effects, and provide conditions for geometric ergodicity and the validity of a central limit theorem.
\end{abstract}


\section{Introduction} 

Despite alternative methods ranging from sequential Monte Carlo \citep{del2006sequential} to variational inference \citep{beal2003variational}, Markov chain Monte Carlo (MCMC) methods remain the default approach among Bayesian statisticians and show no signs of diminishing in importance. 
The overwhelming majority of the literature on MCMC methods has focused on \emph{reversible} Markov chains (that is, Markov chains which satisfy a detailed balance condition), typically constructed as instances of the Metropolis-Hastings (MH) algorithm \citep{metropolis1953equation,hastings1970monte}. This includes MH samplers that obtain efficient joint proposals using gradient information, ranging from Hamiltonian Monte Carlo (HMC, \citealp{duane1987hybrid}) to Metropolis-adjusted Langevin algorithms \citep{roberts1996exponential}. Likewise, this includes the Gibbs sampler  \citep{geman1987stochastic} and generalizations that replace sampling parameters one at a time from their conditional posterior distributions with block updating using a broad class of MH steps.

Data sub-sampling has been explored as a way to speed up MCMC for large datasets \citep{welling2011bayesian, maclaurin2015firefly, quiroz2018speeding}. Sub-samples are used to approximate transition probabilities and reduce bottlenecks in calculating likelihoods and gradients, with the current literature focusing mostly on modifications of the MH algorithm. A major drawback of these approaches is that it is typically difficult to create schemes which preserve the correct target distribution.
While there has been work on quantifying the error for such approximate MCMC schemes \citep{pillai2014ergodicity, johndrow2015approximations, johndrow2017error}, it is in general difficult to do so. 
The pseudo-marginal approach of \cite{andrieu2009pseudo} offers a potential solution, but it is generally impossible to obtain the required unbiased estimators of likelihoods using data sub-samples \citep{jacob2015nonnegative}. 

There is evidence to show that \emph{non-reversible} MCMC methods can offer drastic increased sampling efficiency  over reversible MCMC methods \citep{diaconis2000analysis,sun2010improving,chen2013accelerating, rey2015irreversible}. 
A recently popularized class of non-reversible stochastic processes that can be used to construct sampling algorithms \citep{peters2012rejection,vanetti2017piecewise, fearnhead2018piecewise} are piecewise deterministic Markov processes (PDMPs). PDMPs follow a Markov jump process, where the process evolves deterministically according to some predefined dynamics in between jump events, with the event times being distributed according to a Poisson process. Examples of PDMPs include the bouncy particle sampler (BPS; \citealp{bouchard2018bouncy}) and the zig-zag (ZZ) process \citep{bierkens2019zig}.
Very interestingly, in contrast to traditional MH-based algorithms, PDMPs allow error-free sub-sampling of the data. This remarkable feature has been shown to hold for a wide range of PDMPs \citep{vanetti2017piecewise}.
% , including PDMPs with efficient non-uniform sub-sampling schemes \citep{sen2020efficient}. 

Although theoretically well-founded, PDMP approaches have not yet found widespread use in Bayesian statistics. A major reason for this is the fact that the application of these methods is in general not straightforward. The implementation of PDMPs requires the derivation of upper bounds for the gradient of the log posterior density. These upper bounds must be sufficiently tight for the sampling to remain efficient. While there have been attempts to automate the construction of such upper bounds \citep{pakman2016stochastic} as well as relax the need for upper bounds \citep{cotter2020nuzz}, these lack theoretical guarantees for the exact preservation of the target measure and as such fall into a similar category as approximate MCMC schemes. An additional challenge is that the upper bounds typically deteriorate as the dimension of the parameter space increases, although this can be mitigated to a certain extent using non-uniform sub-sampling schemes \citep{sen2020efficient}.

In this article, we address the problem of increasing the versatility of PDMP-based sampling approaches by introducing a new framework which allows the inclusion of component-wise MCMC updates within a PDMP process. The main idea is to update blocks of components for which efficient upper bounds can be easily derived by a PDMP process, and update blocks of components for which such upper bounds are not easily available with a suitable MH scheme. 
This allows us to combine the versatility of traditional MCMC approaches with the advantages of PDMPs in sampling problems. 
This is particularly relevant to Bayesian hierarchical models, where is it common for certain parameters to have conditional posteriors distributions that are easy to sample from via a Gibbs step, while other parameters can be efficiently updated using a PDMP.
In order to keep the presentation simple and accessible, we focus our attention on the ZZ process in terms of PDMPs, and we refer to our framework as the Gibbs-zig-zag (GZZ) sampler/process. However, we remark that the proposed framework is generic and allows combining a wide class of PDMPs with block-wise MH updates, as presented in \cref{app:Gibbs-PMDP}.

The rest of the article is organised as follows. 
We begin with reviewing the ZZ process in \cref{sec.ZZ.process}.
We present the GZZ sampler in \cref{sec.GZZ}. In particular, we discuss its construction in \cref{sec.gzz}, present its application to posterior sampling from Bayesian hierarchical models in \cref{sec:ap:1}, and summarize its main ergodic properties in \cref{sec:ergodic:summary}. %;  \cref{sec:ergodic:properties:theorems} contains a more detailed discussion of the ergodic properties.}
\cref{sec:appl_GZZ} contains numerical examples for two different contexts related to logistic regression.
Finally, \cref{sec:conclusion} concludes. Proofs and additional details of sampling algorithms are deferred to the appendix. 

\section{The zig-zag sampler}
\label{sec.ZZ.process}

We review the zig-zag (ZZ) process as introduced in \cite{bierkens2019zig} in this section. Consider the problem of sampling  from a probability measure 
\begin{equation*}
\pi( \dd \zeta) 
= 
\frac{1}{Z} {\rm exp}\{-U(\zeta)\} \dd \zeta,
\end{equation*}
where $U\in \C^{2}(\zetaDomain,\RR)$ is a smooth potential function defined on $\zetaDomain \subset \RR^\zetadim$. For the remainder of this paper, we describe the ZZ sampler when $\zetaDomain = \RR^\zetadim$; however, $\zetaDomain$ can be a strict subset of $\RR^\zetadim$ as well \citep{bierkens2018piecewise}.
The ZZ process  $\{\pzeta(t),\pvel(t)\}_{t \geq 0}$ is a piecewise deterministic continuous-time Markov process which lives on an augmented phase space $\zetaDomain \times \{-1,1\}^{\zetadim}$ and is constructed such that the process is ergodic with respect to the product measure $\wt{\pi}(\dd \zeta, \vel) = \pi(\dd \zeta) \mu( \vel)$, where $\mu$ is the uniform measure on $\{-1, 1\}^\zetadim$. The components $\pzeta(t)$ and $\pvel(t)$ are commonly referred to as the position and velocity of the process, respectively. 

For a starting point $\pzeta^{0}$ and initial velocity $\pvel^{0}$, the ZZ process evolves deterministically as 
\begin{equation} \label{eq:det:evol}
\pzeta(t) 
= 
\pzeta^{0} + \pvel^{0} t, \quad \pvel(t) = \pvel^{0}.
\end{equation}
At random times $(\swtime^{k})_{k\in \NN}$, bouncing events occur which flip the sign of one component of the velocity $\pvel^{k-1}$. The process then evolves as \cref{eq:det:evol} with the new velocity until the next change in velocity;  that is, 
\begin{equation}\label{eq:det:evol:2}
\pzeta(\swtime^{k}+s) = \pzeta^{k} + \pvel^{k} s, \quad  \pvel(\swtime^{k}+s) = \pvel^{k},
\end{equation}
 for $s\in [0,\swtime^{k+1}-\swtime^{k}]$, where  $\pvel^{k} = F_{I^{k}}(\pvel^{k-1})$, with random component index $I^{k}$ as specified below and $F_{i}$ denoting the operator which changes the sign of the $i$-th component of its argument, that is $F_{i} : \{ -1,1 \}^{\zetadim} \to \{-1,1 \}^\zetadim$ with $\{F_i(\theta)\}_j = \theta_j$ if $j \neq i$ and $-\theta_j$ if $j=i$. The random event times $(\swtime^{k})_{k \in \NN}$ correspond to arrival times of a non-homogeneous Poisson arrival  process whose intensity function  $m(t) = \sum_{i=1}^\zetadim m_i(t)$ depends on the current phase space value of the process, that is, $m_{i}(t) = \lambda_{i}\{\pzeta(t),\pvel(t)\} ~ (i=1,\dots,\zetadim)$, where $\lambda_1, \dots, \lambda_\zetadim$ are referred to as rate functions. The $k$-th waiting time $\tau^{k} = (\swtime^{k+1} - \swtime^{k})$ of this arrival process is $\tau^{k} = \tau_{I^{k}}^{k}$ with $I^{k} = \argmin_{i \in \{1, \dots, \zetadim\}} \{ \tau_i^{k} \}$, where $\tau_i^{k} ~ (i=1,\dots,\zetadim)$ are random times whose densities are specified by the hazard rates $m^{k}_{i}(s) =  \lambda_{i} \{\pzeta(\swtime^{k}+s), \pvel(\swtime^{k}+s)\}$.

Let $(x)^+ = \max\{0,x\}$ denote the positive part of $x \in \RR$. If the rate functions have the form 
\begin{equation*}
\lambda_{i}(\zeta,\vel)
=
\left \{ \vel_{i} \frac{\partial U(\zeta)}{\partial{\zeta_i}} \right \}^+ + \gamma_{i}(\zeta)
~~ (i=1,\dots,\zetadim)
\end{equation*}
with $\gamma_{i}(\zeta)\geq 0$, this ensures that $\augtarget$ is an invariant measure of the process \citep{bierkens2019zig}, where $\zeta=(\zeta_1,\dots,\zeta_\zetadim)$. The $\gamma_{i}$s are known as the refreshment rates.
Slightly more restrictive conditions ensuring exponential convergence in law to the measure $\augtarget$ and the validity of a central limit theorem can be found in \cite{bierkens2017limit} (see also \citealp{bierkens2019ergodicity}). 

In general, the integrals $\int_{0}^{s}m^{k}_{i}(r) \,\dd r$ of the rate functions $m^{k}_{i}(s)$ do not have a simple closed form, and thus the corresponding first arrival times $\tau_{i}^{k}$ cannot be sampled using a simple inverse transform. Instead, arrival times are usually sampled via a Poisson thinning step \citep{lewis1979simulation} as follows. Assume that we have continuous functions $M_i : \zetaDomain \times \{-1,1\}^{\zetadim} \times \RR_+ \to \RR_+ $ such that $\lambda_{i}(\zeta + s \vel,\vel) \leq M_{i}(\zeta,\vel,s)$. Then
\begin{equation}\label{eq:upper:bound:1}
m^{k}_i( s) 
= 
\lambda_{i}(\pzeta^{k} + s \pvel^{k},\pvel^{k}) \leq M_{i}(\pzeta^{k},\pvel^{k},s) =: M^{k}_i(s)  \quad (i = 1, \dots, \zetadim; ~~ s \geq 0).
\end{equation}
Let $\wt{\tau}^{k}_1, \ldots, \wt{\tau}^{k}_\zetadim$ be the first arrival times of Poisson processes with rates $M^{k}_1(s), \dots, M^{k}_\zetadim(s)$, respectively. Let $I^{k} = \argmin_{ i \in \{ 1, \ldots, \zetadim \} } \{ \wt{\tau}^{k}_i \}$ denote the index of the smallest arrival time. Then, if
\begin{itemize}
\item[(i)] $\pzeta(t)$ is evolved according to \cref{eq:det:evol:2} for time $s=\wt{\tau}_{I^{k}}^{k}$, and % instead of $\tau_{I^{k}}^{k}$, and 
\item[(ii)] after time $\wt{\tau}_{I^{k}}$ the sign of $\vel_{I^{k}}$ is flipped  with probability $m^{k}_{i_0}(\wt{\tau}_{I^{k}}) / M^{k}_{I^{k}}(\wt{\tau}_{I^{k}})$,
\end{itemize}
the resulting process can be shown to be a ZZ process with intensities $m_i(t)=\lambda_i\{\pzeta(t),\pvel(t)\} ~ (i=1,\dots,\zetadim)$ \citep{bierkens2019zig}.

A particularly appealing feature of the ZZ sampler (and PDMPs in general) is that the Poisson thinning procedure can be modified in a way which allows replacing the partial derivatives of the potential function in computations of the event times of bounces by unbiased estimates without changing the invariant measure of the simulated ZZ process \citep{vanetti2017piecewise}. The unbiased estimates can be obtained by sub-sampling of the data when observations are independent.

\section{The Gibbs zig-zag sampler}
\label{sec.GZZ}

\subsection{Process description} \label{sec.gzz}

In practice, derivation of tight upper bounds $M_{i}(t)$ as described in the previous section is often challenging. While using generalized sub-sampling schemes can help in improving the tightness of upper bounds in the setup of sub-sampling \citep{sen2020efficient}, the construction of upper bounds nonetheless remains a fundamental hurdle limiting the use of PDMPs in practice. In order to simplify applications of the ZZ sampler, we introduce a novel extension which combines elements of Gibbs sampling with a PDMP framework. 

Consider a decomposition of the parameter vector as
\begin{equation*}
\zeta = (\pos,\hyp) \in \posDomain \times \hypDomain = \RR^{\posdim}\times\RR^{\hypdim},
\end{equation*}
where $\zetadim = (\posdim + \hypdim)$, and let $\vel \in \{-1,1\}^{\posdim} =: \velDomain$. The idea of the Gibbs zig-zag (GZZ) sampler is to combine updates of the component $\xi$ via a ZZ process, which for fixed value of $\hyp$ preserves the conditional measure
\begin{equation*}
\pi( \dd \pos \mid \hyp) \propto \exp \{-U(\pos,\hyp)\} \,\dd \pos,
\end{equation*}
with conventional (Markov chain) Monte Carlo updates of the second component $\hyp$, which for given value of $\pos$ preserve the conditional measure
\begin{equation*}
\pi( \dd \hyp \mid \pos) 
\propto \exp \{-U(\pos,\hyp) \} \,\dd \hyp.
\end{equation*} 
These updates are combined in such a way that the resulting process is a PDMP which samples the target distribution $\pi$.

More precisely,  let $ \Lc_{{\rm ZZ}}$ denote the generator of the process which leaves the second component $\hyp$ constant while evolving the first component $\pos$ in the corresponding affine subspace according to a ZZ process with rate function
\begin{equation}\label{eq:mod:rate}
\wt{m}_i( t, \hyp)
= 
\left [ \vel_{i}\partial_{\pos_i}U\{\ppos(t),\hyp\} \right ]^+  + \gamma_{i}\{\ppos(t),\hyp\} \quad (i = 1, \dots, \posdim; ~~ t \geq 0);
\end{equation}
we have used the shorthand notation $\partial_{\pos_i} U$ to denote $(\partial/\partial \pos_i) U$.
The generator $\Lc_{{\rm ZZ}}$ takes the form of the differential operator
\begin{equation*}
\left (  \Lc_{{\rm ZZ}}f \right )(\pos,\alpha,\vel) = \sum_{i=1}^{\posdim} \vel_{i} \partial_{\pos_{i}}f(\pos,\alpha,\vel) + \lambda_{i}(\pos,\alpha,\vel) \left [ f\{\pos,\hyp,F_{i}(\vel)\} - f(\pos,\hyp,\vel) \right ], ~~ f\in \test,
\end{equation*}
when considered as an operator on the set of smooth test functions $\test = \C^{\infty}(\domain,\RR)$. 
Here and in the sequel, we consider $\domain = \posDomain \times \hypDomain \times \velDomain$ to be equipped with the product topology induced by the Euclidean norms on  $\posDomain$ and $\hypDomain$, and the discrete topology on $\velDomain$, so that a function $f: \domain \to \RR$ is continuous exactly if $f_{\vel}: (\pos,
\hyp)\mapsto f(\pos,\hyp,\vel)$ is continuous for all $\vel \in \velDomain$. Similarly, we consider the function $f$ to be differentiable if the partial derivatives $\partial_{\pos_{i}}f_{\vel} ~ (i=1,\dots,\posdim)$ and $\partial_{\hyp_{i}}f_\theta ~ (i=1,\dots,\hypdim)$ are well defined for all $\vel \in \velDomain$ and measurable if $f_{\vel}$ is Lebesgue measurable for all $\vel\in \velDomain$; we have used the shorthands $\partial_{\pos_{i}}f_\theta$ and $\partial_{\hyp_{i}}f_\theta$ to denote $(\partial/\partial_{\pos_{i}})f_\theta$ and $(\partial/\partial_{\hyp_{i}})f_\theta$, respectively.

Let $\hyptranskernel$ be a Markov kernel which is such that for any $\pos\in \posDomain$, the conditional measure $\pi( \dd \hyp \mid \pos)$ is preserved under the action of $\hyptranskernel$ in the sense that $\int \hyptranskernel\{(\pos,\hyp^{\prime}),A\} \pi( \dd \hyp^{\prime} \mid \pos) = \int \indicator_{A}(\hyp)\pi ( \dd \hyp \mid \pos)$ for any measurable set $A\subset \hypDomain$, where $\indicator_A(\alpha)$ stands for the indicator function which is such that $\indicator_A(\alpha)=1$ if $\alpha \in A$ and zero otherwise.
Let $(\wt{\swtime}^k)_{k\in \NN}$ denote event times of a Poisson process with constant rate $\swrate>0$. The generator of the PDMP in $\hypDomain$ which is constant in between event times $(\wt{\swtime}^k)_{k\in \NN}$ and whose state is resampled from the Markov kernel $\hyptranskernel$ at event times takes the form $\swrate \Lcg$, where
\begin{equation} \label{eq.gzz_generator}
\left ( \Lcg f\right)(\pos,\alpha,\vel)   = \int_{\hypDomain}\left \{ f(\pos,\hyp^{\prime},\vel) -  f(\pos,\hyp,\vel) \right \}   \hyptranskernel \{ (\hyp,\pos), \dd \hyp^{\prime} \},
\quad
f \in \test.
\end{equation}
We obtain the GZZ process by superimposing the two processes described above; that is, we construct the GZZ process as the process whose generator is %
\begin{equation*}
\Lcgzz 
= 
\Lczz + \swrate \Lcg.
\end{equation*}
The corresponding process $\{ \ppos(t), \pvel(t), \phyp(t)\}_{t \geq 0}$ is a PDMP whose trajectory is piecewise linear in $\pos$ and piecewise constant in $\hyp$. It follows from classical results on the simulation of non-homogeneous Poisson processes that the process can be simulated by generating skeleton points $ \{ (\ppos^{k}, \pvel^{k}, \phyp^{k},\swtime^{k})\}_{k \in \NN}$ according to \cref{alg:Gzigzag} below, which are then linearly interpolated as
\begin{equation} \label{eq:traj:gzz}
\ppos(t) = \ppos^{k} + \pvel^{k} (t - \swtime^{k} ),
\quad 
\phyp(t) = \phyp^{k},
\quad 
\pvel(t) = \pvel^{k},
\quad\text{for} ~ \swtime^{k} \leq t < \swtime^{k+1}.
\end{equation}

\begin{remark}
We mention two recently proposed sampling schemes that  -- similar to the GZZ sampler -- combine ideas of PDMP with aspects of  classical Gibbs sampling, but whose constructions are in fact conceptually very different from the proposed GZZ sampler. 
\begin{enumerate}
\item 
The PDMP proposed in \cite{wu2020coordinate}, termed the coordinate sampler, resembles aspects of a classical Gibbs sampler, but other than the fact that certain components are kept constant in between jump events it bears no resemblance to the GZZ sampler. In terms of its construction, the coordinate sampler falls into the same framework as other popular PDMP processes such as the ZZ sampler \citep{bierkens2019zig} and the BPS \citep{bouchard2018bouncy}, and unlike the GZZ process, it does not allow the incorporation of MH-updates.
\item 
In \cite{zhao2021analysis}, the authors propose a local BPS-within-Gibbs algorithm for the sampling of the posterior distribution of parameters of a continuous Markov chain. Unlike the GZZ sampler, which is a piecewise deterministic continuous time Markov process, the sampling scheme proposed in that work is constructed as a Markov chain where in each step parameters are block-wise updated using either a local BPS sampler or an HMC sampler. A similar construction of a Markov chain that combines a BPS sampler and a Metropolis-Hastings scheme has been also considered in \cite{zhang2021large}. 
\end{enumerate}
\end{remark}

As we discuss in \cref{sec:ergodic:summary}, the GZZ process is path-wise ergodic (see \cref{thm:ergodic}) with respect to the augmented measure 
% \jl{do we need $\otimes$ here?}
% \mscom{done}
% 
\begin{equation}\label{eq:inv:aug}
\wt{\pi}(\dd \zeta, \vel) = \pi(\dd \zeta)\, \mu(\vel), 
\end{equation}
where $\mu$ is the uniform measure on $\{-1,1\}^{\zetadim}$,
under some mild conditions on the potential function $U$. As such, it can be used similarly to other PDMP samplers as a Monte Carlo method for the approximate computation of expectations by finite time trajectory averages, that is,
\begin{equation*}
\EE_{(\pos,\hyp)\sim \target } \left \{ f(\pos,\hyp) \right \} 
\approx
\frac{1}{t} \int_{0}^{t}f \{ \ppos(t),\phyp(t) \} \, \dd t.
\end{equation*}

\begin{algorithm}
\caption{Gibbs zig-zag (GZZ) algorithm.} % for $\swrate = 0$, we formally set $\tau^{\prime} = \infty$ in line 2. } 
\label{alg:Gzigzag} 
\textbf{Input:} $(\ppos^{0},\phyp^{0}, \pvel^{0}) \in \posDomain \times \hypDomain \times \velDomain$% \jl{order switched compared with main text} and $\swtime^{0} =0$.

\begin{algorithmic}[1] 
\FOR{$k = 1, 2, \dots$}
\STATE 
Draw $\tau' \sim {\rm Exponential}(\swrate)$ and $\wt{\tau}_1, \ldots, \wt{\tau}_\posdim$ such that 
\begin{equation}
\PP( \wt{\tau}_i \geq s)
=
\exp \left \{ - \int_0^s \wt{m}_i(T^{k}+r,\phyp^{k}) \, \dd r \right \} 
~ (i = 1, \dots, \posdim).
\end{equation}
\STATE 
Let $\tau^{k} = \min \left \{ \tau', \wt{\tau}_1, \dots, \wt{\tau}_\posdim \right \}$.  
\STATE 
Set $\ppos^{k+1} = \ppos^{k} + \tau^{k} \, \pvel^{k}$ and $\swtime^{k+1} = \swtime^{k} +\tau^{k}$. 
\IF{$\tau = \tau'$}
\STATE 
Draw $\phyp^{k+1} \sim \hyptranskernel\{(\ppos^{k+1},\phyp^{k}),\cdot\}$.
\STATE 
Set $\pvel^{k+1} = \pvel^{k}$. 
\ELSE
\STATE
Set $\phyp^{k+1} = \phyp^{k}$.
\STATE 
Bounce: $\pvel^{k+1} = F_{i_{0}}(\pvel^{k})$, with $i_{0} = \argmin_{i \in \{1,\dots,\posdim\}} \wt{\tau}_{i}$.
\ENDIF
\ENDFOR
\end{algorithmic}
\textbf{Output:} Skeleton points $ \{ (\ppos^{k}, \phyp^{k}, \pvel^{k},\swtime^{k})\}_{k \in \NN}$. 
\medskip  

\end{algorithm}

Practically, ZZ updates of the $\pos$ component can be performed using Poisson thinning. In this case, an upper bound 
$\wt{M}_{i}: \posDomain \times \hypDomain \times \{-1,1\}^{\zetadim} \times \RR_+ \to \RR_+ $ satisfying $\lambda_{i}(\pos + s \vel,\hyp,\vel) \leq \wt{M}_{i}(\pos,\hyp,\vel,s)$ for all $s\geq 0$ is required.

The approach is particularly useful if the restriction of the ZZ process onto the component $\pos$ simplifies construction of upper bounds, and efficient MCMC updates for the remaining component $\hyp$ are available. Such a decomposition is often naturally available in the context of Bayesian posterior distributions with hierarchical priors. We describe the application of the GZZ sampler to such models in \cref{sec:ap:1}. In addition, in the context of Bayesian posterior computation, the GZZ sampler can be modified to allow for sub-sampling of data while exactly preserving the measure $\augtarget$. This is explored numerically in \cref{sec:appl_GZZ}.

\subsection{Bayesian posterior sampling with hierarchical priors}
\label{sec:ap:1}
Hierarchical Bayesian models can often be specified as
\begin{align*}
X_1, \dots, X_\nobs 
& \iid
f(x \mid \pos),
\quad
\pos \mid \hyp 
\sim 
h(\pos \mid \hyp), 
\quad
\hyp 
\sim 
p_0(\hyp),
\end{align*}
where $\hyp$ are hyperparameters with hyper-prior $p_0(\hyp)$, $h$ denotes the conditional distribution of the parameters given the hyperparameters, and $f$ denotes the likelihood of observations given parameters. Bayesian posterior sampling typically proceeds by sampling from 
\begin{equation*}
p(\pos,\hyp \mid X_1, \dots, X_\nobs) 
\propto 
\prod_{j=1}^\nobs f(X_j \mid \pos) \times h(\pos \mid \hyp) \, p_0(\hyp),
\end{equation*}
whose marginal distribution for $\xi$ is the posterior distribution of $\xi$ given $X_1,\dots,X_\nobs$, which is $p(\pos \mid X_1, \dots, X_\nobs) \propto \prod_{j=1}^\nobs f(X_j \mid \pos) \int_{\hypDomain} h(\pos \mid \hyp) \, p_0(\hyp) \, \dd \hyp$.
Letting $\zeta = (\pos,\hyp)$, this corresponds to sampling the Gibbs measure $\pi(\dd \zeta) = Z^{-1} \exp\{-U(\zeta)\} \, \dd \zeta$ with potential function
\begin{equation}\label{eq:decomp}
U(\zeta)
= 
U^0 (\pos,\hyp) + \sum_{j=1}^n U^{j}(\pos), 
\end{equation}
where $U^{0}(\pos,\hyp) = - \log h (\pos \mid \hyp)  - \log p_0 (\hyp)$ and $ U^{j}(\pos) = - \log f(X_j \mid \pos)$. The GZZ sampler can readily be applied in this context, with $\hyptranskernel$ corresponding to either an exact update for the hyperparameters (which is the case when using conditionally conjugate priors), or using a suitable Metropolis--Hastings (MH) scheme such as random walk MH or Hamiltonian Monte Carlo (HMC), when such an exact update is not possible. We consider numerical examples of this in \cref{sec:appl_GZZ}.

\subsection{Ergodic properties}\label{sec:ergodic:summary}
We provide a high level overview of theoretical results on the ergodic properties of the GZZ sampler in this section. Detailed conditions and theorems as well as proofs are deferred to \cref{sec:ergodic:properties:theorems}. These results pertain to the long term properties of trajectory averages 
\begin{equation*} %\label{eq:estimator}
\widehat{\varphi}_{t}
=
\frac{1}{t}\int_{0}^{t} \varphi\{\ppos(s),\pvel(s),\phyp(s)\} \, \dd s
\end{equation*}
of $\augtarget$-integrable functions $\varphi$, and as such are intended to justify the usage of the GZZ sampler as a Monte Carlo method. First, \cref{thm:ergodic} states that under relatively mild conditions (summarized in \cref{as:supp}) on the Markov transition kernel $\hyptranskernel$ and the refreshment rate functions $\lambda_{i} ~(i=1,\dots,\posdim)$, the trajectory averages converge almost surely to corresponding expectations, that is, 
\begin{equation*}
\lim_{t \to \infty} \widehat{\varphi}_{t}
=
\EE_{(\pos,\hyp,\vel)\sim \augtarget}\{\varphi(\pos,\hyp,\vel)\}
~~ \text{almost surely}.
\end{equation*}
Secondly, we show in \cref{thm:ergodic:2} that under additional conditions on the potential function $U$ (see \cref{as:pot}), the GZZ process is geometrically ergodic. This implies a central limit theorem (\cref{thm:clt}) of the form
\begin{equation*}
\sqrt{t} \left [ \widehat{\varphi}_{t}  -  \EE_{(\pos,\hyp,\vel)\sim \augtarget}\{\varphi(\pos,\hyp,\vel)\} \right ] \xrightarrow[t \to \infty]{\mathrm{law}}  \mathcal{N}(0,\sigma^{2}_{\varphi}) \quad \text{for some} ~ \sigma^{2}_{\varphi}>0
\end{equation*}
that holds for a wide class of real-valued functions $\varphi$.

\cref{as:supp} explicitly requires the refreshment rate function $\lambda_{i}$ to be strictly positive. This drastically simplifies the proof of irreducibility of the process. By doing so, we circumvent difficulties (as described and tackled in \citealp{bierkens2019ergodicity} in the case of the standard ZZ process) in the proof of the irreducibility of the process.  From a theoretical perspective, this makes the presented results less interesting. However, we expect that the GZZ process will typically be used in combination with a sub-sampling scheme so that vanishing refreshment rates are unlikely in practice. Likewise, conditions on the potential function $U$ are rather restrictive in terms of the tail properties of the corresponding density and the coupling between the parameters $\pos$ and the hyperparameters $\hyp$. We acknowledge that some of these conditions might not hold in practice.  Instead, we demonstrate in numerical experiments (see \cref{sec:appl_GZZ}) that the central limit theorem remains valid in settings beyond what it is covered by our theoretical results.

\section{Numerical examples} \label{sec:appl_GZZ}

% \subsection{Logistic regression}

Consider the following generic logistic regression model, 
\begin{align} \label{eq.logistic_model}
Y_j 
& \sim 
\mathrm{Bernoulli}\left (\frac1{1+e^{-\psi_j}} \right )
\quad (j = 1, \dots, \nobs),
\end{align}
where $Y_1, \dots, Y_\nobs \in \{0,1\}$ denote observations and $\psi_1, \dots, \psi_\nobs$ are linear predictors that are further assigned a model in a context-specific manner. This is a highly flexible model for which various complexities can be induced by considering different forms for the predictors. PDMP methods for logistic regression with simple (non-hierarchical) priors tend to be efficient \citep{bouchard2018bouncy, bierkens2019zig}, but it is not straightforward to modify these samplers to account for hierarchical structure. 

We run the GZZ sampler for $10^7$ iterations in all our examples. That is, we run \cref{alg:Gzigzag} for $k=1,\dots,10^7$ attempts. We consider sub-sampling, with a sub-sample of size $n_1 < n$ meaning that $n/n_1$ iterations of the GZZ sampler corresponds to approximately one epoch of data evaluation\footnote{In reality, this is actually less than one epoch as updating the hyperparameter $\hyp$ does not involve the data.}, and therefore $10^7$ iterations of GZZ corresponds to $n_1/n \times 10^7$ epochs of data evaluation. We compare the GZZ sampler to HMC-within-Gibbs. We run the Gibbs sampler for a total of $10^4$ iterations for each setting, which means that we make a total of $10^4$ HMC steps as well. For HMC, we consider the leapfrog integrator as described in, for example, Section 2.3 of \cite{neal2011mcmc}. Therefore, for $L$ leapfrog steps, we have $L \times 10^4$ epochs of data evaluation for HMC-within-Gibbs.

\subsection{Random effects model}
\label{sec.random_effects}

Random effects models are routinely applied in a wide variety of disciplines. We consider the following model as illustration,
\begin{align} \label{eq.re_model}
\begin{aligned}
Y_{ij} \mid \beta_j
& \sim 
\mathrm{Bernoulli} \left ( \frac1{1+e^{-\psi_{ij}}} \right ), 
\\
\psi_{ij} 
& = 
m + \beta_j + X_{ij}^\top \upsilon, 
\\
\beta_j 
& \iid 
\N(0, \phi^{-1}),
\end{aligned}
\end{align} 
where $j=1, \dots,K$ index $K$ groups and $i=1, \dots, \nobs$ index $\nobs$ subjects per group\footnote{We assume that each group is of the same size for simplicity; however, this can easily be extended.}, and $\beta_j$ is the random effect for the $j$th group. For the $i$th observation from the $j$th group, 
$Y_{ij} \in \{0,1\}$ denotes the response variable and 
$X_{ij} = (X_{ij1}, \dots, X_{ij\posdim}) \in \RR^\posdim$ denote covariates. In addition, $m$ denotes an overall intercept, $\upsilon = (\upsilon_1, \dots, \upsilon_\posdim)$ denotes the fixed effect coefficients, and $X_{ij}^\top \upsilon = \sum_{l=1}^\posdim X_{ijl} \upsilon_l$. 
We consider the following priors:
\begin{align*}
m 
& \sim
\N(0, \phi^{-1}), \quad 
\upsilon_l 
\iid \N(0, \sigma^2) ~~ (l = 1, \dots, \posdim), 
\\
\phi & \sim \Ga(a_\phi, b_\phi),
\quad
\sigma^2 \sim \IG(a_\sigma,b_\sigma),
\end{align*}
where $\Ga$ denotes a gamma distribution and $\IG$ denotes an inverse-gamma distribution. 
For this problem, we can use a ZZ process with sub-sampling to update $(\upsilon, m,\beta_1, \dots, \beta_K)$ conditionally on the hyperparameters $(\phi,\sigma^2)$, while the conditional distributions for the hyperparameters can be exactly sampled from; details are provided in \cref{app:mixed_eff}. In the notation of \cref{sec.gzz}, we have $\pos = (\upsilon_1, \dots, \upsilon_\posdim, m, \beta_1, \dots, \beta_K)$ and $\alpha=(\phi, \sigma^2)$. 

We consider synthetic data generated from model \eqref{eq.re_model} with true $(m, \delta, \pos) = (m_{\text{true}}, \delta_{\text{true}}, \pos_{\text{true}}) \in \RR^{1+K+\posdim}$. The covariates $X_{ijl}$s are sampled from the mixture distribution $\covdist_\sparsity(\dd x) = \sparsity \delta_0 (\dd x) + (1-\sparsity) \rho(\dd x)$, where $\delta_0(\dd x)$ is a point mass at zero, $\rho$ is a standard normal density, and $\sparsity \in (0,1]$ denotes the level of sparsity among the covariates.   

In a first experiment, we study the effect of the switching rate $\swrate$ on the mixing of process; recall that this is given in \cref{eq.gzz_generator}. To this end, we consider a simple setup with $n=10$, $K=2$, and $\posdim=2$, and we also choose $\sparsity = 0.5$. 
We run the GZZ sampler for various values of the switching rate for mini-batch size ten.
We plot the integrated auto-correlation time of the slowest component of $\xi$ in the left panel of \cref{fig.mixing_random_effects}.
The mixing improves to a certain point as the switching rate increases, beyond which the improvement tapers off. In particular, for $\eta \leq 10^{-1}$ the integrated auto-correlation time of the slowest component is approximately proportional to $\eta^{-1}$.

In another experiment, we compare the mixing of the process to the size of the mini-batch used. This is shown in the right panel of \cref{fig.mixing_random_effects}.
When the switching rate is low, increasing the mini-batch size does not have a noticeable effect on the mixing of the process. However, when the switching rate is in the ``flat'' part of the left panel of \cref{fig.mixing_random_effects} (that is, $\swrate = 6.47$), increasing the mini-batch size has a clear effect on the mixing of the process. 

\begin{remark}
We obtain the integrated auto-correlation times as follows. We first extract equally-spaced samples from the continuous-time process obtained by running the GZZ sampler. We plot each component of $\xi$'s auto-correlation function and calculate its integrated auto-correlation time by observing when the auto-correlation function converges to zero and then summing the auto-correlation function up to that time. We do not display these plots here, but these can be found in the notebooks accompanying our code.
\end{remark}

\begin{figure}
\centering 
\includegraphics[width=0.44\textwidth]{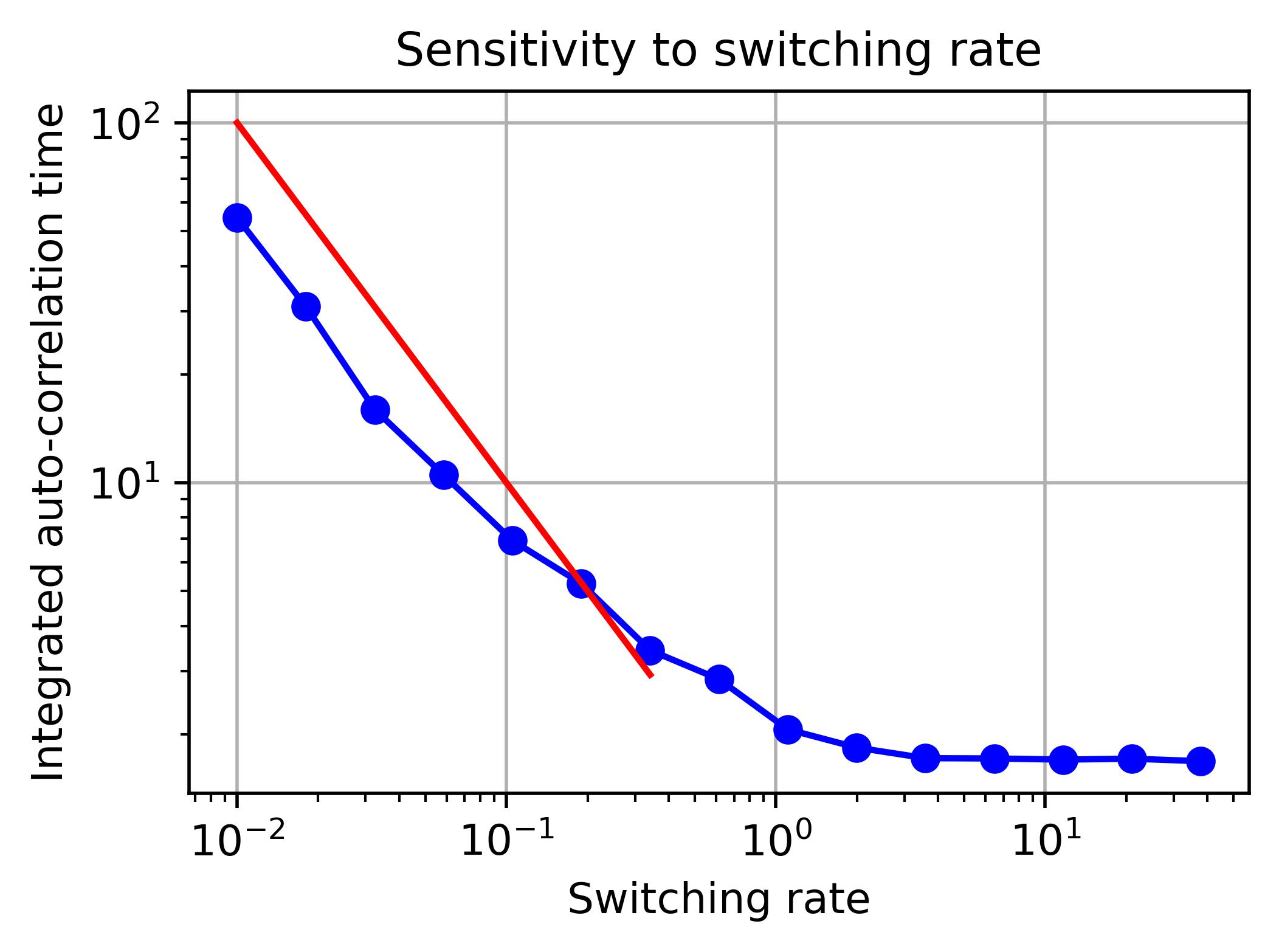}
\includegraphics[width=0.55\textwidth]{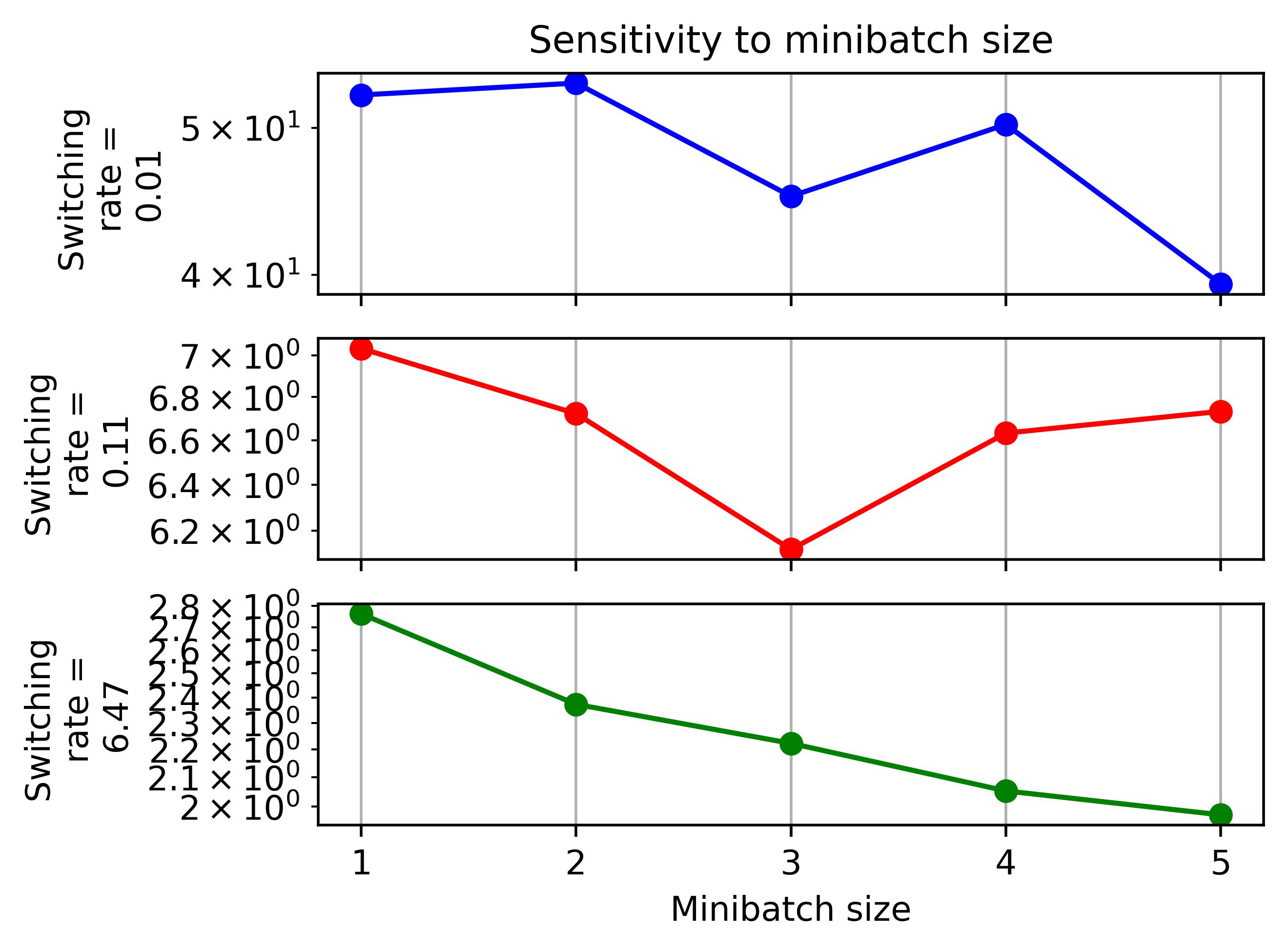}
\caption{Sensitivity to switching rate $\swrate$ and mini-batch size for the random effects model. The red line in the left plot shows the graph of $\eta^{-1}$.}
\label{fig.mixing_random_effects}
\end{figure}

Next, we compare the GZZ sampler to HMC-within-Gibbs.
We choose $\sparsity = 5 \times 10^{-2}$, which means that the covariates are $95\%$ sparse. For HMC-within-Gibbs, we replace the ZZ updating by HMC. 
In this case, we choose $n=100$ and $\posdim=5$, and vary the number of groups $K$. As $K$ increases, both the dimension of the sampling problem $(1+K+\posdim)$ as well as the total number of observations $K \times \nobs$ increases. 
We tune HMC by choosing a range of different leapfrog steps and stepsizes, and looking at cases where the acceptance rate is close to the optimal acceptance rate of $0.651$ \citep{beskos2013optimal}. Among them, we choose the combination of step-size and number of leapfrog steps which gives the highest effective sample size per epoch of data evaluation.
We plot the relative effective sample size per epoch of data evaluation for GZZ with sub-sampling divided by the same for HMC in \cref{fig.randomeffects_GZZvsHMC}, where we observe that the relative performance of using GZZ over HMC increases as the number of groups increases.

\begin{remark}
The total number of observations $K n$ increases as the number of groups $K$ increases, which makes each iteration of HMC slower. However, since GZZ uses sub-sampling, the time per iteration of GZZ remains the same. The number of random effects increases as the number of groups increases, which means that the dimension of the parameter space increases (if we treat the random effects as ``parameters'' whose posterior is to be sampled from). While it is true that GZZ performs worse as the dimension increases, this is offset by the slower run time of HMC.
\end{remark}

\begin{figure}
\centering 
\includegraphics[width=0.6\textwidth]{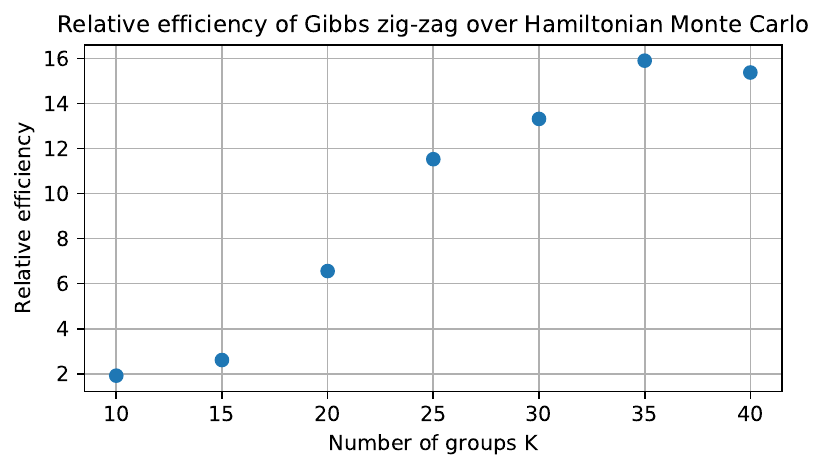}
\caption{Comparison of effective sample size per epoch of data evaluation for Gibbs zig-zag with sub-sampling and Hamiltonian Monte Carlo for the random effects model.}
\label{fig.randomeffects_GZZvsHMC}
\end{figure}

We acknowledge that the assumptions made in order to prove the central limit theorem (\cref{thm:clt}), namely Assumptions \ref{as:supp} and \ref{as:pot}, do not hold for this example. We nevertheless demonstrate numerically that a central limit theorem does appear to hold for this example. To this end, we choose a simple setting with $K=5$ groups and $n=50$ observations per group, and we choose $p=10$ covariates. We let the test function $\varphi$ be simply the identity function and estimate $\varphi_\true := \EE_{(\pos,\hyp,\vel)\sim \augtarget}\{\varphi(\pos,\hyp,\vel)\}$ by running \cref{alg:Gzigzag} for a very long time. 
We then choose several different time horizons $T_1 < \cdots < T_M$ and run the GZZ sampler up to each $T_m$; this can be achieved by running \cref{alg:Gzigzag} till the total time $T^k$ reaches $T_m$. We run the GZZ sampler independently $R = 2 \times 10^2$ times for each $T_m$ and obtain estimates $\widehat{\varphi}_{T_m, \, r}$ for $m=1, \dots, M$ and $r=1,\dots,R$. \cref{thm:clt} implies that $\{ T_m^{1/2} ( \widehat{\varphi}_{T_m,r}  - \varphi_\true ) \}_{r=1}^R$ should converge to samples from a Gaussian distribution with zero mean as $m$ increases. We demonstrate this by noting that if $X \sim \N_p(\mu, \Sigma)$, then $(X-\mu)^\top \Sigma^{-1} (X-\mu) \sim \chi^2_p$. This means that $\{ [ T_m^{1/2} ( \widehat{\varphi}_{T_m,r} - \varphi_\true ) ]^\top \widehat{\Sigma}_m^{-1} [ T_m^{1/2} ( \widehat{\varphi}_{T_m,r} - \varphi_\true ) ] \}_{r=1}^R$ should converge to samples from a $\chi^2_p$ distribution as $m$ increases, where $\widehat{\Sigma}_m$ is the empirical covariance matrix of $\{ T_m^{1/2} [ \widehat{\varphi}_{T_m,r}  - \varphi_\true ] \}_{r=1}^R$.
We display QQ plots for this in \cref{fig.CLT-mixed_effects}, where we observe that a central limit theorem appears to hold in this setting.

\begin{figure}
\centering 
\includegraphics[width=\textwidth]{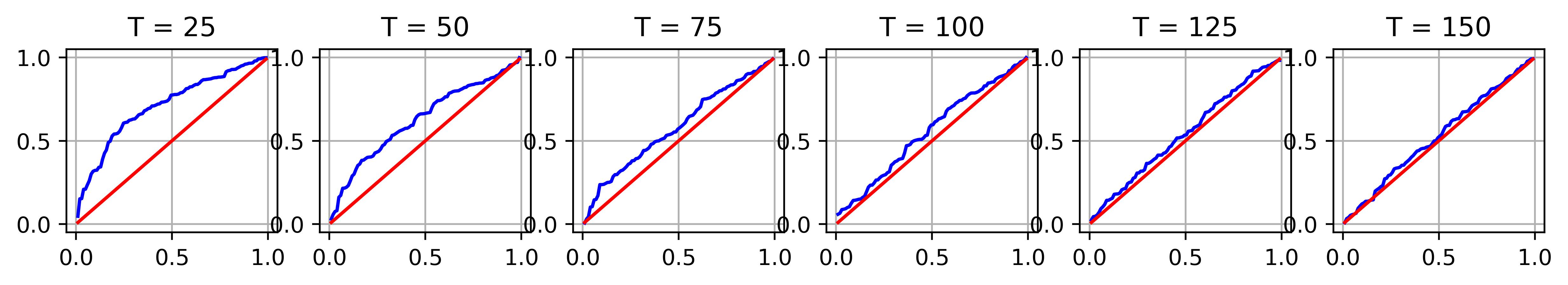}
\caption{QQ plots for the mixed effects model providing empirical evidence for a central limit theorem.}
\label{fig.CLT-mixed_effects}
\end{figure}

\subsection{Shrinkage prior}
\label{sec.shrinkage_prior}

Consider the case where we have $\posdim$ covariates and let $X_j =  (X_{j1}, \dots, X_{j \posdim})$ be the covariates for the $j$th observation $Y_j$. Equation \eqref{eq.logistic_model} then corresponds to a typical logistic regression model with $\psi_j = \upsilon_0 + \sum_{i=1}^\posdim X_{ji} \upsilon_j$, 
where $\upsilon = (\upsilon_1, \dots, \upsilon_\posdim)$ are coefficients for the covariates and $\upsilon_0$ is an intercept term. 
Even when $\posdim$ is relatively small compared to $\nobs$, the posterior for $\upsilon$ is not concentrated around a reference point if the covariates are sparse and the prior is isotropic Gaussian. 
We instead use the GZZ sampler to employ a shrinkage prior for $\upsilon$. A popular shrinkage prior is the spike-and-slab prior \citep{mitchell1988bayesian, ishwaran2005spike}, which is a mixture of a spike at zero and a higher variance component. We consider the following specification of the spike-and-slab prior:
\begin{align*}
\upsilon_i 
& \ind 
\gamma_i \, \N ( 0, \nu \tau_i^2) + (1-\gamma_i) \, \N(0, \tau_i^2 ), 
\\
\gamma_i
& \iid
\textrm{Bernoulli}(\pi)
~~ (i = 1, \dots, \posdim),
\\
\nu & \sim \IG(a_\nu, b_\nu), 
\quad
\pi \sim \mathrm{Beta}(a_\pi, b_\pi),
\end{align*}
where $\gamma_i \in \{0,1\}$ $(i = 1, \dots, \posdim)$, and we choose $\upsilon_0 \sim \N(0,\sigma_0^2)$ for the intercept. 
Finally, as recommended by \cite{polson2012half}, we choose i.i.d. half-Cauchy priors for the $\tau_i$s as $p_0(\tau_i) \propto ( 1 + \tau_i^2/d_\tau )^{-(d_\tau+1)/2}$ $(i=1,\dots,p)$.
In terms of the notation of \cref{sec.gzz}, a ZZ process with sub-sampling can be used to update $\pos = (\upsilon_0, \dots, \upsilon_\posdim)$ conditionally on the hyperparameters $\alpha = (\gamma_1, \dots, \gamma_\posdim, \tau^2, \pi, \nu)$, while the conditional distributions for the hyperparameters can be sampled using MCMC update steps. Details are provided in \cref{app:spikeslab_conditionals}. In contrast to the random effects model of \cref{sec.random_effects}, the dimension of the hyperparameter $\hyp$ is more than twice that of the parameter $\pos$ in this case. 
We consider synthetic data with the covariates being generated in the same way in \cref{sec.random_effects}.  The responses $Y_i$ are sampled from model \eqref{eq.logistic_model} with ``true'' $\upsilon = \upsilon_{\text{true}} \in \RR^{\posdim+1}$.

In a first experiment, we study the effect of varying mini-batch sizes and varying switching rates $\swrate$ on the efficiency of the GZZ sampler. We choose a simple example with $n=50$ and $\posdim=20$, and $\sparsity = 0.4$, and we make the ``true'' $(\upsilon_1,\dots,\upsilon_\posdim)$ sparse by setting only $20\%$ of its components to be non-zero. 
We run the GZZ sampler for various values of the switching rate for mini-batch size ten and plot the integrated auto-correlation time of the slowest component of $\xi$ in the left panel of \cref{fig.mixing_shrinkageprior}. 
and look at the sensitivity to the mini-batch size in the right panel of \cref{fig.mixing_shrinkageprior}. The results are similar to those in \cref{sec.random_effects}. In particular, the mixing improves to a certain point with increasing switching rate, beyond which it tapers off, and for $\eta \leq 10^{-1}$ the integrated auto-correlation time of the slowest component is again approximately proportional to $\eta^{-1}$. 
Increasing the mini-batch size does not have a noticeable effect on the mixing for low switching rates and has a clear effect when the switching rate is sufficiently high.
\begin{figure}
\centering 
\includegraphics[width=0.44\textwidth]{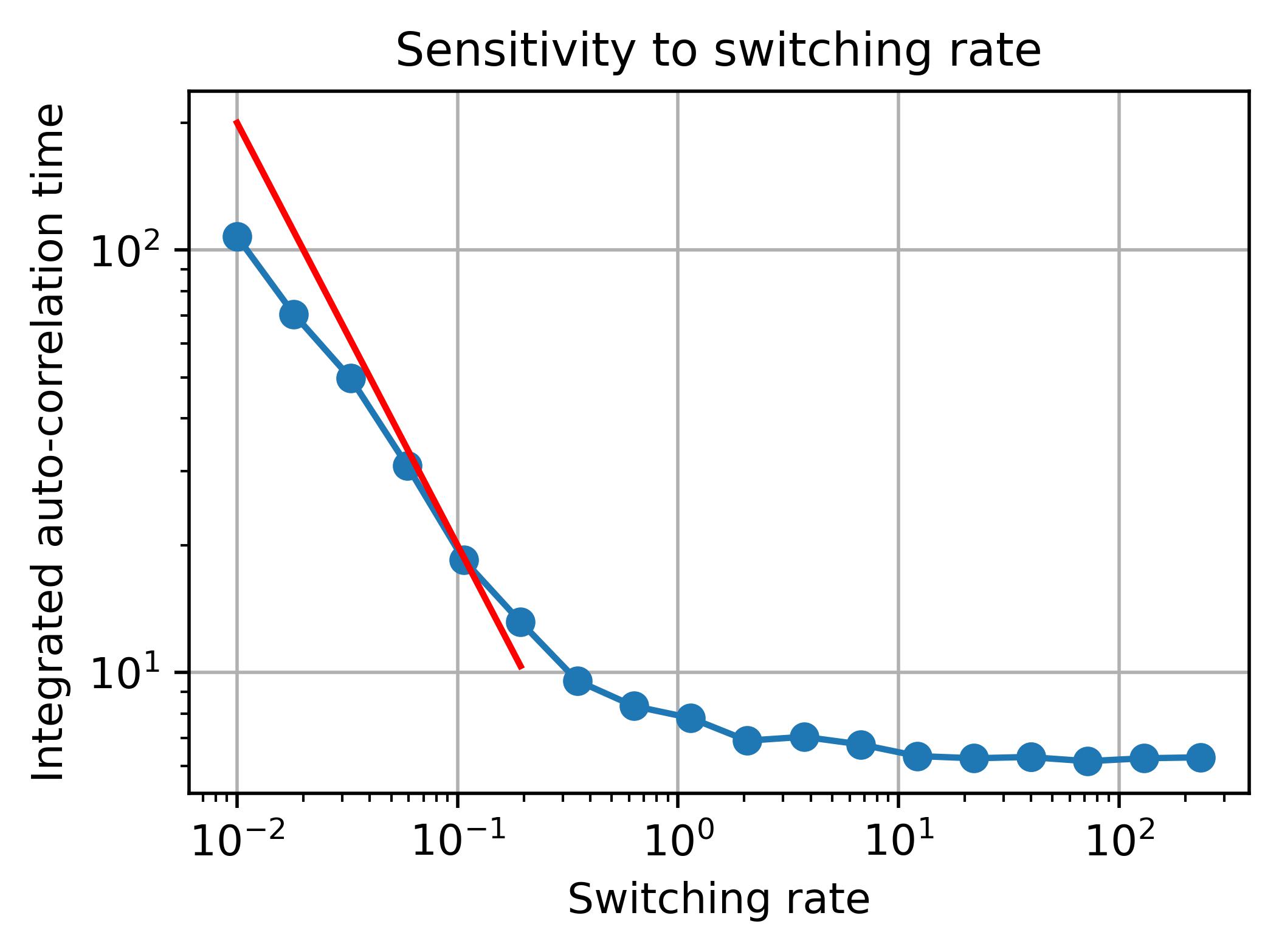}
\includegraphics[width=0.55\textwidth]{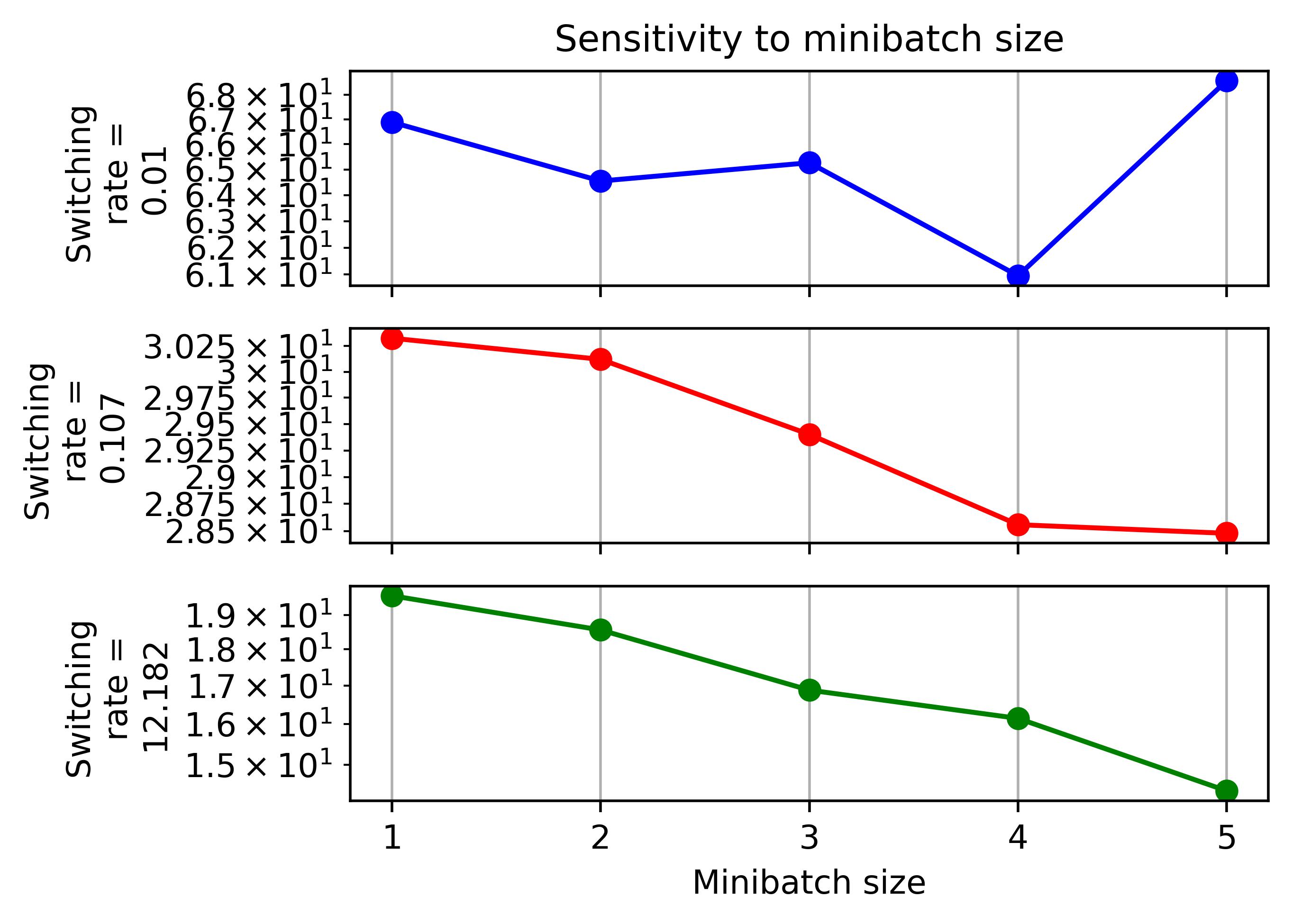}
\caption{Sensitivity to switching rate $\swrate$ and mini-batch size for logistic regression with spike-and-slab prior. The red line in the left plot shows the graph of $\eta^{-1}$}.
\label{fig.mixing_shrinkageprior}
\end{figure}

We compare the GZZ sampler to HMC-within-Gibbs. We consider $\posdim=10^2$ and varying values of $n$. We also choose the ``true'' $(\upsilon_1,\dots,\upsilon_\posdim)$ to be sparse with only 10\% of its components being non-zero. For each value of $n$, we choose $\sparsity$ such that $\sparsity \times n$ is fixed at $50$.
We tune HMC in the same way as in \cref{sec.random_effects} and compare the effective sample size per epoch of data evaluation of the GZZ sampler with sub-sampling and HMC in the left plot of \cref{fig.shrinkage_GZZvsHMC}. We observe that as $n$ increases, the GZZ sampler improves upon HMC; this is due to GZZ being faster due to sub-sampling.

We also perform experiments where we fix $n=10^2$ and consider increasing values of $p$. The ``true'' $(\upsilon_1,\dots,\upsilon_\posdim)$ is again chosen to be sparse with only 10\% of its components being non-zero for each value of $p$. The right plot of \cref{fig.shrinkage_GZZvsHMC} displays the effective sample size per epoch of data evaluation of the GZZ sampler with sub-sampling as compared to HMC. We observe that HMC becomes more efficient as compared to GZZ as the dimension $p$ increases.

\begin{figure}
\centering 
\includegraphics[width=0.49\textwidth]{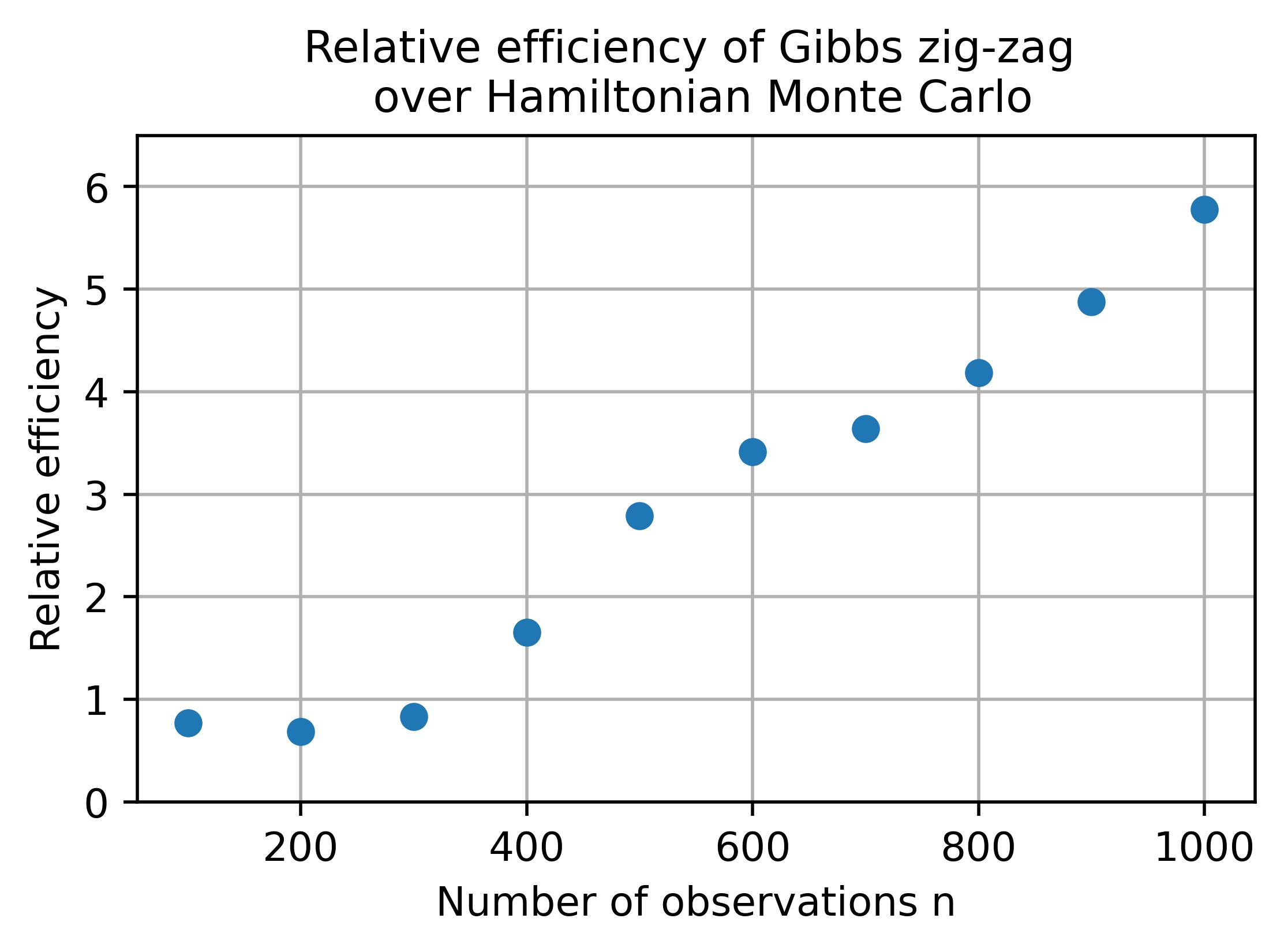}
\includegraphics[width=0.49\textwidth]{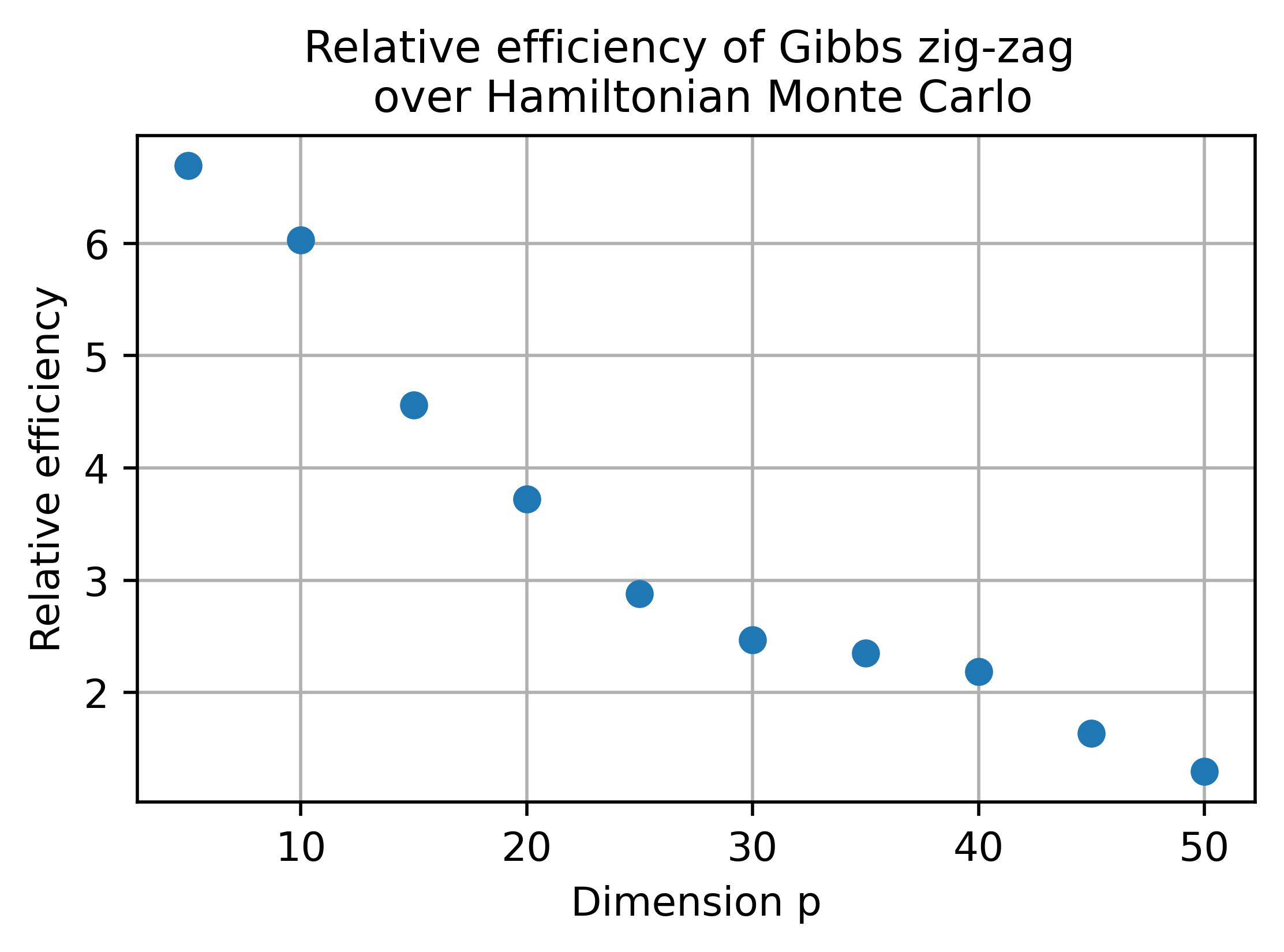}
\caption{Comparison of effective sample size per epoch of data evaluation for Gibbs zig-zag with sub-sampling and Hamiltonian Monte Carlo while using a shrinkage prior; the left plot is for $p=10^2$, and the right plot is for $n=10^2$.}
\label{fig.shrinkage_GZZvsHMC}
\end{figure}

\subsection{Choice of rate parameter  $\eta$}

In the above reported numerical experiments, we observed that the integrated auto-correlation time of the slowest mixing component is monotonically decreasing, approximately proportional to $\eta^{-1}$ for sufficiently small $\eta$, and approximately constant for large values of $\eta$. This observation is consistent with the large-deviation results for similar systems (see \citealp{lu2019methodological}) and can be used to derive the following heuristic for the parametrization of $\eta$. 

Let $C_{\rm Gibbs}$ and $C_{\rm ZZ}$ denote the total computational time incurred for updates of the hyperparameters $\alpha$ and the parameters $\pos$ during a simulation of the GZZ process, respectively. We suggest choosing $\eta$ such that $C_{\rm Gibbs}$ and $C_{\rm ZZ}$ are of comparable magnitude (for example, such that $r_{\rm Gibbs} = { C_{\rm Gibbs}}{(C_{\rm Gibbs}+C_{\rm ZZ})^{-1}} \approx 0.2$). This can be easily achieved in practice since $C_{\rm Gibbs} \propto \eta T$, and thus the computational time associated with Gibbs updates can be easily controlled by varying the value of $\eta$.

The motivation behind the heuristic is as follows. Let $\eta_{{\rm opt}}>0$ denote a value of $\eta$ that results in maximal sampling efficiency (measured in terms of effective sample size of the slowest mixing component per computational cost). 
If the determined value of $\eta$ is larger than $\eta_{{\rm opt}}$, then it follows that the loss in sampling efficiency relative to an optimal choice of $\eta$ is bounded from above by $r_{\rm Gibbs}$. Otherwise, if the determined value of $\eta$ is smaller than $\eta_{\rm opt}$, then, since the decrease of the integrated auto-correlation time is at most proportional to $\eta^{-1}$, a further increase of $\eta$ would not result in significant increase of sampling efficiency.

\begin{figure}
\centering 
\includegraphics[width=\textwidth]{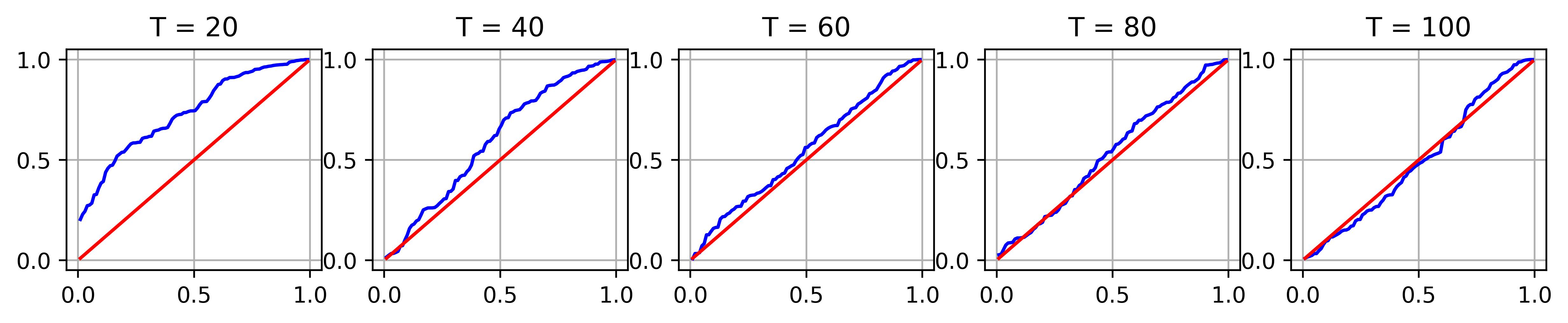}
\caption{QQ plots for the spike-and-slab prior.}
\label{fig.CLT-shrinkage_prior}
\end{figure}

\section{Discussion} 
\label{sec:conclusion}

Piecewise deterministic Markov process (PDMP) methods present a promising alternative to traditional (reversible) MCMC algorithms for sampling from posteriors in Bayesian inference. In this paper, we have combined one of the popular PDMPs, the zig-zag process, with Gibbs-like updates. Other variants of the framework that incorporate different PDMPs can be straightforwardly implemented as well (see \cref{app:Gibbs-PMDP}).
PMDP-based sampling schemes have found limited applications in past years, mainly due to the fact that for many sampling problems the construction of suitable tight upper bound is cumbersome, if not impossible. This includes the type of posterior sampling problems with hierarchical priors considered in this article. Thus, the proposed framework contributes to extending the applicability of PDMP-based sampling.

In terms of performance, our approach inherits both positive and negative features of classical PDMP-based approaches. Exact sub-sampling allows for computationally very efficient and asymptotically exact sampling in the presence of large data. However, as also demonstrated for other PDMP sampling methods (see, for example, \citealp{quiroz2021block}), sampling efficiency in comparison to HMC-within-Gibbs tends to deteriorate as the dimensionality of the sampling problem increases relative to the number of observations.

There are many interesting follow-up directions. While we have focused on PDMP schemes that preserve the exact target distribution, it could be useful to combine Gibbs-like updates with PDMP schemes that only approximately preserve the target distribution like those in \cite{pakman2017binary, cotter2020nuzz}. Theoretically, it would be interesting to study high-dimensional scaling limits of the GZZ process along the lines of \cite{bierkens2022high, deligiannidis2021randomized}. Moreover, the derivation of $\swrate$-dependent spectral estimates for the generator of the GZZ process using the Hypocoercivity framework by \cite{dolbeault2015hypocoercivity} (see \cite{andrieu2021hypocoercivity} for an adoption of that framework to PDMPs) would be of interest in order to gain a better understanding of the effect of parameter choices for $\swrate$ on sampling efficiency of the GZZ sampler. 

\section*{Acknowledgment}

DS and DD acknowledge support from National Science Foundation grant 1546130.
MS and DS acknowledge support from grant DMS-1638521 from SAMSI. The work of JL is supported in part by the National Science Foundation via grants DMS-1454939 and CCF-1934964 (Duke TRIPODS).

\bibliographystyle{apalike}
\bibliography{references}

\newpage 
\appendix

\section{Ergodic properties and central limit theorem}\label{sec:ergodic:properties:theorems}

\subsection{Additional notations}

In the following, we use $\PP_{(\pos,\hyp,\vel)}( \cdot ) = \PP\{\cdot  \mid  [\ppos(0),\phyp(0),\pvel(0)]=(\pos,\hyp,\vel) \} $ as a shorthand for probabilities in terms of the path measure of the GZZ process with initial value $(\pos,\hyp,\vel)$. Similarly, we use the shorthand $\EE_{(\pos,\hyp,\vel)}( \cdot ) = \EE\{\cdot  \mid  [\ppos(0),\phyp(0),\pvel(0)]=(\pos,\hyp,\vel) \} $ for expectations with respect to the same path measure. 
Moreover, for given $t\geq 0$ we denote by 
\begin{equation*}
\transkernel_{t}\{(\pos,\hyp,\vel),\cdot\} 
=
\PP_{(\pos,\hyp,\vel)} [ \{ \ppos(t),\pvel(t),\phyp(t)\} \in \cdot],
\end{equation*}
the transition kernel associated with the GZZ process. The transition kernel $\transkernel_{t}$ may be considered as an operator on the set of probability measures on $\domain$ whose action on a probability measure $\nu$ is defined as
\begin{equation*}
(\transkernel_{t} \nu)(\cdot) = \sum_{\vel \in \velDomain} \int_{\posDomain\times \hypDomain} \transkernel_{t}\{(\pos,\hyp,\vel),\cdot\} \nu(\dd\pos\,\dd \hyp, \vel).
\end{equation*}

\subsection{Invariant measure and Harris recurrence} 

We first assert that $\augtarget$ is indeed an invariant measure of the GZZ process.
\begin{proposition}\label{prop:inv}
The GZZ process has $\augtarget(\dd \pos\,\dd \hyp, \vel)$ as an invariant measure; that is, $\transkernel_{t} \augtarget = \augtarget$ for $t\geq 0$.
\end{proposition}
 
In order to show that the invariant measure $\augtarget$ is unique, we require the following assumption pertaining to the transition kernel $\hyptranskernel$ for hyperparameter updates and the switching rates $\lambda_{i} ~(i=1,\dots,\posdim)$. 
\begin{assumption}[on $\hyptranskernel$ and $\lambda_{i} ~(i=1,\dots,\posdim)$]\label{as:supp} \

\begin{enumerate}[(A)]
\item \label{as:supp:it:1}
The Markov transition kernel $\hyptranskernel$ possesses a smooth density, and for any $(\pos,\hyp)\in \hypDomain$, its associated probability measure has full support on $\hypDomain$,  that is,
\begin{equation*}
\hyptranskernel \left \{ ( \pos,\hyp),A \right \}  
=
\int_{A} \hyptransdens\{ (\pos,\hyp), \hyp^{\prime} \} \, \dd \hyp^{\prime}, 
\end{equation*}
with $\hyptransdens\in \C^{\infty} \left [ (\posDomain \times \hypDomain) \times \hypDomain,  (0,\infty)\right ]$ and $\hyptransdens\{(\pos,\hyp),\cdot\}>0$ for all $(\pos,\hyp) \in \posDomain\times\hypDomain$ and all measurable sets $A\subset \hypDomain$. 
\item \label{as:supp:it:2}
The switching rates are bounded away from zero, that is, there exists $\lambdamin>0$ such that $\lambda_{i}(\pos,\hyp,\vel) \geq \lambdamin$ for all $i=1,\dots,\posdim$ and all $(\pos,\hyp,\vel) \in \domain$.
\end{enumerate}
\end{assumption}
For $\augtarget$-integrable $\varphi$, let
\begin{equation*} %\label{eq:estimator}
\widehat{\varphi}_{t}
=
\frac{1}{t}\int_{0}^{t} \varphi\{\ppos(s),\pvel(s),\phyp(s)\} \, \dd s,
\end{equation*}
as previously defined, be the corresponding finite trajectory average up to time $t$. Uniqueness of the invariant measure as well as some other regularity properties of the GZZ process, which hold under \cref{as:supp}, ensure that a law of large numbers (path-wise ergodicity) holds for $\widehat{\varphi}_{t} $ as $t\to \infty$. This is made precise in the following theorem.
\begin{theorem}\label{thm:ergodic}
If \cref{as:supp} is satisfied, then the GZZ process is ergodic with unique invariant measure $\augtarget$. In particular, the process is path-wise ergodic in the sense that
\begin{equation*}
\lim_{t \to \infty} \widehat{\varphi}_{t}
=
\EE_{(\pos,\hyp,\vel)\sim \augtarget}\{\varphi(\pos,\hyp,\vel)\}
~~ \text{almost surely}
\end{equation*}
for any real-valued $\augtarget$-integrable test function $\varphi$.
\end{theorem}

\subsection{Geometric ergodicity and central limit theorem} 
In addition to path-wise ergodicity of the process, we show exponential convergence (geometric ergodicity) of the GZZ process for the practically relevant case where updates of the hyperparameters are performed as Gibbs updates, that is, $\hyptranskernel\{(\pos,\hyp),\cdot \} = \target( \cdot \mid \pos)$.

More precisely, we show exponential decay of the semi-group operators $(\exp({t\Lcgzz}))_{t\geq 0}$ in a suitable weighted $L^{\infty}$-space as $t\to \infty$, where
\begin{equation*}
\{\exp(t\Lcgzz)\varphi\}(\pos,\hyp,\vel) 
=
\EE_{(\pos,\hyp,\vel)}[\varphi\{\ppos(t),\phyp(t),\pvel(t)\}]    
\end{equation*}
denotes the evolution operator associated with the GZZ process. In order for exponential convergence to hold we require the potential function $\pot$ to satisfy certain asymptotic growth conditions and we require the excess switching rates $\gamma_{i} ~ (i=1,\dots,\posdim)$ to be bounded.

\begin{assumption}[On potential function $U$ and excess switching rates $\gamma_{i}$]\label{as:pot}\
\begin{enumerate}[(A)]
\item \label{as:pot:it:1} 
%The potential function  $\pot$ satisfies the asymptotic growth conditions
There exist continuous functions $g_{i} : \posDomain \rightarrow [0,\infty)~ (i=1,2)$, satisfying $g_{i}(\pos) \rightarrow 0$ as $\abs{\pos} \rightarrow \infty$ and a constant $c>0$ so that the inequalities
\begin{equation}
\frac{\max \{ 1,\norm{ {\rm Hess}_{\pos} \pot(\pos,\hyp)} \} }{|\nabla_{\pos} \pot(\pos,\hyp)|}
\leq g_{1}(\pos) 
\quad
\text{and}
\quad 
\frac{\abs{\nabla_{\pos} \pot(\pos,\hyp) }}{\pot(\pos,\hyp)} \leq g_{2}(\pos),
\end{equation}
hold for all $\hyp \in \hypDomain$ and $\pos \in \posDomain$ with $\abs{\pos}>c$.
Here ${\rm Hess}_{\pos} \pot$ and $\nabla_{\pos}\pot$ denote the Hessian and gradient of the function $\pos \mapsto U(\pos,\hyp)$,
respectively, and $\abs{\cdot}$ and $\norm{\cdot}$ denote the Euclidean norm and the Frobenius norm, respectively.
\item \label{as:pot:it:3}
The excess switching rates $\gamma_{i} ~ (i=1,\dots,\posdim)$ are bounded from above, that is, there exists $\gammamax>0$ so that
\begin{equation}\label{eq:def:lya}
\sup_{(\pos,\hyp)\in \posDomain\times \hypDomain} \gamma_{i}(\pos,\hyp) \leq \gammamax.
\end{equation}
\item \label{as:pot:it:2} 
Let $\delta>0$ and $a>0$ be such that $0\leq \gammamax \delta <a  < 1$ with $\gammamax$ as specified in \cref{as:pot}. Define the function 
\begin{equation*}
V(\pos,\hyp,\vel)  
=
\exp \left [ a \pot(\pos,\hyp) + \sum_{i=1}^{\posdim} \phi \left \{ \vel_{i} \partial_{\pos_{i}}\pot(\pos,\hyp) \right \} \right ]
\end{equation*}
where $\phi(s) = {\rm sign}(s) \log(1+\delta \abs{s})/2$. There exist a choice of $a$ and $\delta$, and a constants $r>0$ and $c>0$ such that the inequality 
\begin{equation} \label{eq:lya:eqn:1}
\int_{\hypDomain} \frac{V(\pos,\aatt,\vel)}{V(\pos,\aa,\vel) } \exp\{-\pot(\pos,\aatt)\} \dd \aatt + r
<
\int_{\hypDomain} \exp\{-\pot(\pos,\aatt)\} \dd \aatt
\end{equation}
holds for all $(\pos,\aa) \in \posDomain \times \hypDomain$ with $\abs{(\pos,\aa)}>c$, and all $\vel \in \{-1,1\}^{\posdim}$.
\end{enumerate}
\end{assumption} 
The following theorem provides a simple (yet restrictive) condition on the form of the potential energy function $U(\pos,\hyp,\vel)$ which is sufficient for \cref{as:pot}\eqref{as:pot:it:2}, to be satisfied.

\begin{theorem}\label{thm:cond:pot}
\cref{as:pot}\eqref{as:pot:it:2} is satisfied if \cref{as:supp} and  \cref{as:pot}\eqref{as:pot:it:3}, hold, and the potential function $U$ can be decomposed as $U(\pos,\hyp)= U_{1}(\pos) + \bfunc(\pos,\hyp) + U_{2}(\hyp)$, where  $\bfunc$ is such that 
the absolute values of $\bfunc$ and its derivatives are bounded, that is, there exists $\bmax >0$ such that 
\begin{equation*}
\abs{\bfunc(\pos,\hyp)} \leq \bmax 
\quad \text{and} \quad
\abs{ \partial_{\pos_{i}} \bfunc(\pos,\hyp)} \leq \bmax 
\end{equation*}
for all $(\pos,\hyp) \in \posDomain \times \hypDomain$, and $i = 1, \dots, \posdim$. 
\end{theorem}
For $V : \domain \to [1,\infty)$, define the corresponding weighted $L^{\infty}$-norm as 
\begin{equation*}
\norm{\varphi}_{L^{\infty}_{V}}
= 
\norm*{\frac{\varphi}{V}}_{L^{\infty}}, \quad \varphi : \domain \to \RR ~ \text{measurable},
\end{equation*}
and denote by $L^{\infty}_{V}(\domain)$  the Banach space induced by this norm. 
Under \cref{as:supp,as:pot}, there exists a suitable function $V$ such that the difference between $\exp(t\Lcgzz)\varphi$ and the expected value of $\varphi$ under the target measure decays exponentially ${L^{\infty}_{V}(\domain)}$ as $t\to\infty$.
\begin{theorem} \label{thm:ergodic:2}
Let \cref{as:supp,as:pot} be satisfied and $\hyptranskernel\{(\pos,\hyp),\cdot \} = \target( \cdot \mid \pos)$, and consider the function $V$ as defined in \eqref{eq:def:lya}.

There exist $\const >0$ and $\lambda>0$ such that
\begin{equation}\label{eq:L2:exp:conv}
\forall \, t \geq 0 ~\text{and}~ \forall \, \varphi \in L^{\infty}_{V}(\domain), \quad \norm*{e^{t\Lcgzz}\varphi - \int_\domain \varphi \,\dd \pi}_{L^{\infty}_{V}} \leq \const \,e^{-t\lambda} \norm*{ \varphi - \int_\domain \varphi \,\dd \pi}_{L^{\infty}_{V}}.
\end{equation}
\end{theorem}
We prove \cref{thm:ergodic:2} using Lyapunov techniques as presented in, for example, \cite{meyn2012markov}. More specifically, we show the result as a consequence of Theorem 3.4 of \cite{hairer2011yet} by demonstrating that (i) $V$ satisfies a {\em Lyapunov condition} of the form
\begin{equation}\label{eq:lya:1}
\Lcgzz V  \leq - a\, V + b\,\indicator_{C},
\end{equation}
where  $a>0,b\in \RR$ are constants and $C$ is a compact set, and (ii) the process satisfies a {\em minorization condition} of the form
\begin{equation}\label{eq:minorization}
\forall (\pos,\hyp,\vel) \in C,
\quad 
\PP_{(\pos,\hyp,\vel)} \left [ \{ \ppos(t), \pvel(t),\phyp(t)\} \in \cdot \right ]
\geq 
c \,\lebegue( \cdot \cap C),
\end{equation}
where $\lebegue$ denotes the Lebesgue measure and $C\subset\domain$ is the same compact subset as in the Lyapunov condition \eqref{eq:lya:1}.

Let $L^{\infty}_{V,0}(\domain)$ denote the subspace of $L^{\infty}_{V}(\domain)$ which is comprised of test functions with vanishing expectation, that is, $L^{\infty}_{V,0}(\domain)= \{ \varphi \in L^{\infty}_{V} : \EE_{(\pos,\hyp,\vel)\sim\augtarget}\{\varphi(\pos,\hyp,\vel)\} = 0\}$. 
\cref{thm:ergodic:2} implies (\citealp[Proposition 2.1]{lelievre2016partial}) directly the following corollary.

\begin{corollary}
\label{col:invert}
Under the same conditions as of \cref{thm:ergodic:2}, the operator $\Lcgzz$ considered on $L^{\infty}_{V,0}(\domain)$ is invertible, and
\begin{equation*}
\Lcgzz^{-1} = - \int_{0}^{\infty} e^{t\Lcgzz} \dd t
\quad \text{and} \quad 
\norm*{\Lcgzz^{-1}}_{\B(L^{\infty}_{V,0})} \leq \frac{ \const}{\lambda},
 \end{equation*}
where $\const$ and $\lambda$ are the same constants as in \cref{thm:ergodic:2}, $\|\G\|_{\B(L^{\infty}_{V,0})} = \sup_{f\in L^{\infty}_{V,0}}  (\| Gf \|_{L^{\infty}_{V,0}}) / (\|f\|_{L^{\infty}_{V,0}})$ denotes the operator norm induced by $\|\cdot\|_{L^{\infty}_{V,0}}$, and $\B(L^{\infty}_{V,0})$ denotes the space of bounded linear operators on $L^{\infty}_{V,0}$ on which the operator norm is well-defined.
\end{corollary}

By \cite{bhattacharya1982functional} and the boundedness of the inverse of the generator, a central limit theorem is obtained as follows.

\begin{corollary}[Central limit theorem for GZZ]\label{thm:clt}
Consider the setup of \cref{thm:ergodic:2} and let $\varphi \in L^{\infty}_{V}(\domain)$. Then there exists $\sigma^{2}_{\varphi}>0$ so that 
\begin{equation*}
\sqrt{t} \left [ \widehat{\varphi}_{t}  -  \EE_{(\pos,\hyp,\vel)\sim \augtarget}\{\varphi(\pos,\hyp,\vel)\} \right ] \xrightarrow[t \to \infty]{\mathrm{law}}  \mathcal{N}(0,\sigma^{2}_{\varphi}).
\end{equation*}
\end{corollary}

\section{Proofs}\label{sec:app:proofs}
\subsection{Proof of Proposition \ref{prop:inv}}

\begin{proof}
It is sufficient to show that $\sum_{\vel \in \{-1,1\}^{\posdim}} \int_{\posDomain} \int_{\hypDomain}\left ( \A f\right)(\pos,\alpha,\vel) \pi(\dd \alpha, \dd \pos) = 0$
for $\A \in \{ \Lczz, \Lcg \}$. 
For any value of $\hyp$, it can be shown that\footnote{see \citealp[Theorem 2.2]{bierkens2019zig} for detailed calculations.} 
\begin{equation*}
\sum_{\vel \in \{-1,1\}}^{\posdim}\int_{\posDomain} \left (  \Lczz f \right )(\pos,\alpha,\vel) \, \pi(\dd \pos \mid \alpha) 
=
0,
\end{equation*}
and thus in particular 
\begin{align*}
& \quad 
\sum_{\vel \in \velDomain}\int_{\posDomain} \int_{\hypDomain} \left (  \Lczz f \right )(\pos,\alpha,\vel) \pi(\dd \pos\, \dd \alpha)
\\
& = 
\int_{\hypDomain}\sum_{\vel \in \velDomain}\int_{\posDomain} \left (  \Lczz f \right )(\pos,\alpha,\vel) \pi(\dd \pos \mid \alpha) \pi(\dd \alpha)
=
0,
\end{align*}
which proves that $\Lczz$ preserves the target measure. Similarly, for any value of $\pos\in \posDomain,\vel \in \{-1,1\}^{\posdim}$,
\begin{align*}
& \quad 
\int_{\hypDomain}\left ( \Lcg f\right)(\pos,\alpha,\vel) \pi(\dd \alpha \mid \pos)  
\\
& =
\int_{\hypDomain}\int_{\hypDomain}\left \{ f(\pos,\hyp^{\prime},\vel) -  f(\pos,\hyp,\vel) \right \}   \hyptranskernel \{ (\pos,\hyp),  \dd \hyp^{\prime}\} \pi(\dd \alpha \mid \pos)
\\
& = 
\int_{\hypDomain}f(\pos,\alpha^{\prime},\vel) \left [ \int_{\hypDomain} \hyptranskernel \{ (\pos,\hyp), \dd \hyp^{\prime}\} \pi(\dd \alpha \mid \pos)  \right ]
\\
& \quad  -
\int_{\hypDomain} f(\pos,\alpha,\vel) \left [ \int_{\hypDomain}\hyptranskernel \{ (\pos,\hyp),  \dd \hyp^{\prime}\} \right ] \pi(\dd \alpha \mid \pos) 
\\
& = 
\int_{\hypDomain} f(\pos,\alpha^{\prime},\vel)\pi(\dd \alpha^{\prime} \mid \pos) 
-
\int_{\hypDomain} f(\pos,\alpha,\vel) \pi(\dd \alpha \mid \pos) 
= 0,
\end{align*}
where the second-to-last equality follows from the fact that $\hyptranskernel$ is a transition kernel which preserves the conditional measure $\pi(\dd \alpha \mid \pos)$. The proof is concluded as
\begin{align*}
& \quad 
\sum_{\vel \in \velDomain}\int_{\posDomain} \int_{\hypDomain}\left ( \Lcg f\right)(\pos,\alpha,\vel) \pi(\dd \alpha\, \dd \pos)
\\
& = 
\sum_{\vel \in \velDomain} \int_{\posDomain} \int_{\hypDomain}\left ( \Lcg f\right)(\pos,\alpha,\vel) \pi(\dd \alpha \mid \pos) \pi(\dd \pos)
= 
0. 
\end{align*}
\end{proof}

\subsection{Additional notations}

For convenience purposes, we extend the definition of the flip operator $F_{i}$ to index values $i \in \{0,\dots,\posdim+1\}$ as follows: if $i\in \{1,\dots,\posdim\}$, we let $F_{i}$ to be defined as in \cref{sec.ZZ.process}, and if $i \in \{0,\posdim+1\}$, we define $F_{i}$ to simply be the identity map. Moreover, for a $k$-tuple $(i_{1},\dots,i_{k})$, we let $F_{(i_{1},\dots,i_{k})} = F_{i_{k}} \circ \cdots \circ F_{i_{1}}$ denote the concatenation of the corresponding flip operators. 

We refer to a tuple $\ubf = (\tbf,\ibf)$, where $\tbf = (t_{1},\dots,t_{m+1}) \in (0,\infty)^{m+1}$ with $0<t_{1}<\dots<t_{m+1}$ and $\ibf = (i_{1},\dots,i_{m}) \in \{1,\dots,\posdim\}^{m}$ for some $m\in \NN$, as a {\em control sequence}.  The control sequence defines a piecewise linear trajectory on the time interval $[0,t_{m+1}]$ as follows: 
\begin{align*}
\vel(t) 
& =
F_{(i_{1},\dots,i_{k})}\vel ~~ \text{if} ~ t_{k}\leq t < t_{k+1} ~~ (k= 0,\dots,m),
\\
\pos(t) 
& =
\pos + \int_{0}^{t}\pos(s) \dd s.
\end{align*}
We use $\Phi_{\ubf}(\pos,\vel)= \{\pos(t_{m+1}),\vel(t_{m+1})\}$ as a shorthand notation for the final position of the trajectory. In the PDMP literature, a control sequence $(\tbf,\ibf)$ is said to be {\em admissible} if the  rates $\lambda_{i_{k}}$ in a vicinity of each point of the corresponding trajectory at times $t_{k} ~ (k=1,\dots m)$ are positive. Note that in the setup considered in this article, we do not require a generalization of the concept of admissibility of a control sequence, since the rates $\lambda_{i}(\pos,\hyp) ~ (i=1,\dots,\posdim)$ are by \cref{as:supp}\eqref{as:supp:it:2} always positive irrespective of the value of the hyperparameter. In particular, since the support of the marginal of $\target$ in $\pos$  is a connected set,  it follows that for any pair of points  $(\pos,\hyp,\vel), (\tpos,\thyp,\tvel) \in \domain$, there exists an admissible control sequence $\ubf$ such that $\Phi_{\ubf}(\pos,\vel) = (\tpos,\tvel)$ irrespective of the values of $\hyp$ and $\thyp$.

\subsection{Poisson thinning procedure}\label{sec:pt}

In the proofs of the following lemmata, we repeatedly use a Poisson thinning procedure for the simulation of a restricted version of the GZZ process up to a prescribed finite time $\tmax>0$.
The procedure is akin to \cref{alg:Gzigzag}. However, we constrain hyperparameter values to a compact set $\hypcompact \subset \hypDomain$, so that for prescribed $\pos\in \posDomain$ and any realization of the GZZ process with $\ppos(0)=\pos$ and $[\phyp(s)]_{s \leq \tmax} \subseteq \hypcompact$,
\begin{equation}\label{eq:upper:bound}
\lambdamax = \max \left \{ \lambda_{i}(\tpos, \vel, \hyp) :  \tpos \in \ball_{\tmax}(\pos), ~(\vel \in \velDomain; ~\thyp \in\hypcompact; ~ i=1,\dots,d+1) \right \} 
\end{equation}
is an upper bound of the rate function values $\lambda_{i}\{\ppos(t),\phyp(t),\pvel(t)\}$ up to time  $\tmax$. 
Here and in the sequel, we denote the constant function $(\pos, \hyp, \vel) \mapsto \swrate$ by $\lambda_{d+1}$. Constraining the hyperparameter values as described above allows us to apply a Poisson thinning procedure as follows. Arrival times $\ta^{k} ~ (k=1,2,\dots)$ are sampled from a Poisson process with constant rate $(\posdim+1)\overline{\lambda}$. For each arrival time, a component index $I_{k}$ is sampled uniformly from the set $\{1,\dots,d+1\}$ and a uniform random variable is simulated as $U_{k}\sim {\rm Uniform}([0,1])$. Skeleton points are generated sequentially as $\pppos^{k+1} = \pppos^{k} + (\ta^{k+1}-\ta^{k})\ppvel^{k}$, and by applying 
an accept/reject step as follows.
\begin{itemize}

\item 
If $U_{k} \leq {\lambda_{I_{k}}(\pppos^{k},\ppvel^{k},\pphyp^{k})}/{\lambdamax}$, then either the $I_{k}$th velocity component is flipped, that is, $\ppvel^{k+1} = F_{I_{k}}(\ppvel^{k})$ if $1\leq I_{k} \leq d$, or, if $I_{k}=d+1$, the hyperparameter block is updated as  $\pphyp^{k+1} \sim \tilde{\hyptranskernel}_{\hypcompact}\{ (\pppos^{k+1},\pphyp^{k}), \cdot \}$,
where $\tilde{\hyptranskernel}_{\hypcompact} \{ (\pppos^{k+1},\pphyp^{k}), \dd \hyp^{\prime} \} = Z^{-1} \hyptransdens \{ (\pppos^{k+1},\pphyp^{k}), \hyp^{\prime} \}  \indicator_{\hypcompact}(\hyp^{\prime})\dd \hyp^{\prime}$.

\item 
If $U_{k} > {\lambda_{I_{k}}(\pppos^{k},\ppvel^{k},\pphyp^{k})}/{\lambdamax}$, then $\ppvel^{k+1}=\ppvel^{k},\pphyp^{k+1}=\pphyp^{k}$, and $I_{k}$ is set to zero indicating a rejection event.

\end{itemize}
By interpolating the generated skeleton points as specified in \cref{eq:traj:gzz}, the obtained process \\ $[\pppos(t),\pphyp(t),\ppvel(t)]_{t \in [0,\tmax] }$ 
is identical in law to the GZZ process on $[0,\tmax]$ which targets the  probability distribution 
\begin{equation*}
\augtarget_{\hypcompact}(\dd \pos\,\dd \hyp, \vel) = {Z_{\hypcompact}^{-1}}{\rm exp}\{-U(\pos,\hyp)\}  \indicator_{\hypcompact}(\hyp)\mu(\vel) \, \dd \pos \, \dd \hyp
\end{equation*}
where $Z_{\hypcompact}$ is a suitable normalization constant.

\subsection{Proof of Theorem \ref{thm:ergodic}} \label{sec:abs:cont}

Recall that $\swtime^{k} ~ (k \in \NN)$ denote the random times at which either components of the velocity are flipped or the hyperparameters are updated. Let $N$ denote the random integer which is such that $\swtime^{N}$ is the first time when (i) the hyperparameters have been updated, and (ii) $(\posdim-1)$ distinct components of $\vel$ have been flipped. If this does not occur, we set $N=\infty$. Moreover, we let $\tnp=\swtime^{N+1}$ provided that $N<\infty$, and $\tnp= \infty$ otherwise, so that $\tnp$ can be understood as the first event time after both the hyperparameter block has been updated and at least $(\posdim-1)$ distinct components of the velocity have been switched. The following lemma states that the law of $(\ppos(\tnp),\phyp(\tnp))$ is absolutely continuous with respect to the Lebesgue measure. 

\begin{lemma}\label{lem:abs:cont}
Let \cref{as:supp} be satisfied. Then $\PP_{(\pos,\hyp,\vel)}\{\tnp<\infty, (\ppos(\tnp),\phyp(\tnp)) \in B\} = 0$ for any $(\pos,\hyp,\vel) \in \domain$ and any measurable set $B \subset \posDomain \times \hypDomain$ with Lebesgue measure zero.
\end{lemma}

\begin{proof}
Let $B$ be a measurable set of Lebesgue measure zero in $\posDomain \times \hypDomain$, and $\tmax\geq 0$ be arbitrary. For a prescribed $\hyp\in \hypDomain$ and $\delta>0$, let $\ball_{\delta}(\hyp) = \{ \thyp \in \hypDomain : \abs{\thyp-\hyp} \leq \delta\}$ be the closed ball of radius $\delta$ centered at $\hyp$, and $\event_{\tmax} = \{ [\phyp(s)]_{s \leq \tmax} \subseteq \ball_{\tmax}(\hyp) \}$ denote the event that up to time $\tmax$, the hyperparameter component of the GZZ process remains within the ball $\ball_{\tmax}(\hyp)$. In order to prove the lemma, it suffices to show that 
\begin{equation}\label{eq:zero:L}
\PP_{(\pos,\hyp,\vel)} \left [\left \{ \tnp<\tmax \right \} \cap \event_{\tmax}  \cap \left \{ [\ppos(\tnp),\phyp(\tnp)] \in B \right \} \right ] 
=
0,
\end{equation}
as this implies the statement of the lemma in the limit $\tmax\to \infty$ by monotone convergence. 

When constrained to realizations in $\event_{\tmax}$, the GZZ process is identical in law to the process $(\pppos(t),\pphyp(t),\ppvel(t))_{t \in [0,\tmax] }$ generated by the thinning procedure described in \cref{sec:pt}. Using the notation introduced there, we can write $\pppos(\tau)$ as
\begin{equation*}
\pppos(\tau) 
=
\pos + \tau^{1}\vel + \tau^{2}F_{I_{1}}\vel + \dots + \tau^{M+1}F_{I_{1},\dots,I_{M}}\vel,
\end{equation*}
where $M\geq N$ is a random integer, and $\tau^{k} = (\ta^{k} -\ta^{k-1}) ~ (k\in \NN \cup \{0\})$ with $\ta^{0}=0$ denoting the  inter-arrival (waiting times) of the Poisson process.
Let $R_{m}$ denote the set of indices $(i_{1},\dots,i_{m}) \in \{0,\dots,d+1\}^{m}$ which are such that $\posdim$ different indices appear in $(i_{1},\dots,i_{m}) $ and at least one of them is $(\posdim+1)$. Moreover, let
\begin{equation*}
\const
= 
\max_{ \pos^{\prime} \in  \ball_{\tmax}(\pos),\, \tthyp \in \ball_{\tmax}(\hyp), \,  \hyp^{\prime} \in \ball_{\tmax}(\hyp) \, }  
Z^{-1} \hyptransdens \left [ (\pos^{\prime},\tthyp), \hyp^{\prime} \right ]  
\indicator_{\hypcompact}(\hyp^{\prime})
\, 
\lebegue\left \{ \left [ \ball_{\tmax}(\hyp) \right] \right \},
\end{equation*}
and let $U\sim {\rm Uniform}\left [ \ball_{\tmax}(\hyp) \right]$ be a uniform random variable independent of the inter-arrival times $\tau_{i} ~ (i \in \NN)$. Then, 
\begin{align}
& \quad 
\PP_{(\pos,\hyp,\vel)} \left ( \event_{\tmax} \cap \left \{ \tau<\tmax,   [\ppos(\tau),\phyp(\tau)] \in B  \right \} \right )
=
\PP_{(\pos,\hyp,\vel)} \left [ \tau<\tmax, \{\pppos(\tau),\pphyp(\tau)\} \in B \right ]
\nonumber \\
& \leq
\sum_{m \in \NN} \sum_{(i_{1},\dots,i_{m}) \in R_{m}}
\hspace{-.2cm}\PP_{(\pos,\hyp,\vel)} \left [ \event_{\tmax} \cap  \left \{ \tau<\tmax,\, \left[\pos + \tau^{1}\vel + \dots + \tau^{m+1}F_{i_{1},\dots,i_{m}}\vel, \pphyp(\tau)  \right ] \in B \right \} \right ] 
\nonumber \\
& \leq
\sum_{m \in \NN} \sum_{(i_{1},\dots,i_{m}) \in R_{m}}
\const \, \PP_{(\pos,\hyp,\vel)} \left [ \event_{\tmax} \cap  \left \{   \tau<\tmax,\, \left(\pos + \tau^{1}\vel + \dots + \tau^{m+1}F_{i_{1},\dots,i_{m}}\vel, U \right ) \in B \right \} \right ] 
\label{ineq:lem:abs:cont}
\end{align}
For each term in $(\pos + \tau^{1}\vel + \dots + \tau^{m+1}F_{i_{1},\dots,i_{m}})$, the vectors $(\vel,F_{i_{1}}\vel,\dots,F_{i_{1},\dots,i_{m}} \vel)$ span $\RR^{\posdim}$, and $\tau^{k} ~ (k \in \NN)$ are independent exponentially distributed random variables.  Similarly, the law of $U$ is absolutely continuous with respect to the Lebesgue measure and $U$ is independent of $\tau^{k} ~(k \in \NN)$. 
Thus, the distribution of $( \pos + \tau^{1}\vel + \dots + \tau^{M+1}F_{I_{1},\dots,I_{M}}\vel, U)$ is absolutely continuous with respect to the Lebesgue measure on $\posDomain\times \hypDomain$.
This implies that all probability terms in the sum of \eqref{ineq:lem:abs:cont} are zero since $B$ is assumed to be a set of zero Lebesgue measure in $\posDomain\times \hypDomain$.
\end{proof}

\begin{lemma}[Continuous component]\label{lem:cont:fix}
For any two points $(\pos, \hyp, \vel)\in \domain$ and $(\ttpos,\ttvel, \tthyp)\in \domain$, there exist open sets $\posset,\hypset, \ttposset$ and $\tthypset$, with $\pos \in \posset, \hyp \in \hypset, \ttpos \in \ttposset$, and $\tthyp \in \tthypset$, and constants $\varepsilon>0, t^{\prime}>0,c>0$ such that for any $\tpos \in \posset, \thyp \in \hypset$ and all $t \in (t^{\prime}, t^{\prime}+\varepsilon]$, 
\begin{equation}\label{eq:cont:comp}
\PP_{\tpos,\vel,\thyp} \left [ \ppos(t)\in \tposset, \pvel(t) = \ttvel,\phyp(t) \in \thypset \right ]
\geq 
c \,\lebegue( \tposset \cap \ttposset) \, \lebegue(\thypset \cap \tthypset), 
\end{equation}
for any Borel-measurable sets $\tposset\subset \posDomain$ and $\thypset\subset \hypDomain$. 
\end{lemma}

\begin{proof}
Let $\tposset \subseteq \posDomain, ~ \thypset \subseteq \hypDomain$, and $\ttvel \in \velDomain$ be arbitrary and $B = \tposset \times \thypset \times \{\ttvel\}$. Consider a control sequence $\ubf = (\tbf,\ibf) = (t_{1},\dots,t_{m+1}; i_{1},\dots,i_{m})$ which is such that $\Phi_{\ubf}(\pos,\vel)=(\ttpos,\ttvel)$ and all component indices appear at least once in $\ibf$. Let $\hypcompact$ denote a compact set whose interior contains both $\hyp$ and $\tthyp$. Let $\tmax = (t_{m+1}+1)$ and $\event_{\tmax} = \{  [\phyp(s)]_{s \leq \tmax} \subseteq \hypcompact \} $. We have,
\begin{equation*}
\PP_{(\pos,\hyp,\vel)} \left \{ [\ppos(t),\phyp(t),\pvel(t)] \in B \right \} 
\geq
\PP_{(\pos,\hyp,\vel)} \left [ \left \{ [\ppos(t),\phyp(t),\pvel(t)] \in B \right\} \cap \event_{\tmax} \right ]
\end{equation*}
for any $t\geq 0$. By constraining the process to realizations contained in $\event_{\tmax}$, we can again use the Poisson thinning procedure described in \cref{sec:pt} to simulate the law of the GZZ process up to time $\tmax$. Now consider a collection of closed, bounded and disjoint intervals $\neighbour_{1}, \dots, \neighbour_m$ which are neighborhoods of the points $t_{1},\dots,t_{m}$, respectively, and $\neighbour_{m}$ is such that for sufficiently small $\varepsilon>0$, the interval $\neighbour_{m+1} = [\max\, \neighbour_{m} + \varepsilon,t_{m+1}]$ has non-empty interior. For the equivalent process generated by the Poisson thinning procedure, consider the event $\event = \event_{1} \cap \event_{2}$, where $\event_{1}$ is the event that $\ta^{k}\in \neighbour_{k} \, \forall \, k=1,\dots,m$, and $\event_{2}$ is the event that $I_{k}=i_{k} \, \forall \, k=1,\dots,m$ and $I_{m+1}=\posdim+1$, (which in particular implies that all velocity flips and the update of the hyperparameter block are accepted).
Then for $t = t_{m+1}$, we have
\begin{align*}
& \quad 
\PP_{(\pos,\hyp,\vel)} \left [ \{ \ppos(t), \pvel(t), \phyp(t) \} \in B \right ]
\geq 
\PP_{(\pos,\hyp,\vel)} \left [ \left \{ [\ppos(t),\pvel(t),\phyp(t)] \in B \right\} \cap \event_{\tmax} \right ] 
\\
& =
\PP \left [ \{\pppos(t),\ppvel(t),\pphyp(t)\} \in B \right ]
\geq 
\PP \left [ \left \{ [\pppos(t),\ppvel(t),\pphyp(t)] \in B   \right \} \cap \event \right ]
\\ 
& =
\PP  \left [ \left [ \left \{  \Psi(\pos,\hyp, t, \ta^{1}, \ta^{2},\dots, \ta^{m}), F_{i_{1},\dots,i_{m}}\vel,\pphyp^{m+1} \right \} \in B  \right ] \cap \event \right ]
\\
& =
\PP \left [ \left \{  \Psi(\pos,\hyp, t, \ta^{1}, \ta^{2},\dots, \ta^{m}) \in  \tposset \right \} \cap \left \{ \pphyp^{m+1} \in \thypset \right \} \cap \event\right ],
\end{align*}
where
\begin{equation*}
\Psi(\pos,t, s_{1}, s_{2},\dots, s_{m} )
=
\pos + s_{1}\vel + (s_{2}-s_{1})F_{i_{1}}\vel + \dots + (t-s_{m}) F_{i_{1},\dots,i_{m}}\vel.
\end{equation*}
Let 
\begin{equation*}
\lambdamin
= 
\min \left \{ \lambda_{i}(\tpos, \vel, \hyp) :  \tpos \in \ball_{\tmax}(\pos), ~\vel \in \velDomain, ~\thyp \in\hypcompact, ~i=1,\dots,\posdim+1 \right \}.
\end{equation*}
Using standard results on Poisson processes and the fact that the random variables $E_{k}, U_{k} ~ (k=1,\dots,m+1)$ are mutually independent, we find 
\begin{align*}
& \quad 
\PP(\event) = \PP(\event_{2} \mid \event_{1}) \PP(\event_{1})
\\
& \geq
\left \{ \frac{ \lambdamin}{(\posdim+1)\lambdamax }   \right \}^{m+1}\exp \left \{-  \overline{\lambda} \left ( \tmax - \sum_{k=1}^{m+1} \abs{\neighbour_{i}} \right ) \right \} \prod_{i=1}^{m+1}\overline{\lambda} \, \abs{\neighbour_{i}} \exp\{-\overline{\lambda} \, \abs{\neighbour_{i}}\}
=
\const_{1},
\end{align*}
and therefore
\begin{align*}
& \quad
\PP \left [  \left \{  \Psi(\pos, t, \ta^{1},\dots, \ta^{m}) \in   \tposset  \right \} \cap \left \{ \pphyp^{m+1} \in \thypset \right \}  \cap \event \right ]
\\
& =
\PP \left [ \left \{  \Psi(\pos, t, \ta^{1},\dots, \ta^{m}) \in   \tposset  \right \} \cap \left \{ \pphyp^{m+1} \in \thypset \right \} \mid \event   \right ] \PP(\event)
\\
& \geq  
\const_{1} \PP \left [ \left \{  \Psi(\pos, t, \ta^{1},\dots, \ta^{m}) \in   \tposset  \right \} \cap \left \{ \pphyp^{m+1} \in \thypset \right \} \mid \event  \right ].
\end{align*}
Conditioning on $\event=\event_{1} \cap \event_{2}$ renders the arrival times $\ta^{1}, \dots, \ta^{m+1}$ to be mutually independent random variables with supports $\neighbour_{1}, \dots, \neighbour_{m+1}$, respectively. Thus, there exists $\const_{2}>0$ such that
\begin{align*}
& \quad
\PP \left [ \{  \Psi(\pos, t, \ta^{1},\dots, \ta^{m}) \in   \tposset  \} \cap \left \{ \pphyp^{m+1} \in \thypset \right \} \mid \event  \right ]
\\
& \geq 
\const_{2}\,\PP \left [\{ \Psi(\pos, t, \wt{U}^{1}, \dots,\wt{U}^{m}) \in   \tposset \} \cap \{ \pphyp^{m+1} \in \thypset =\} \mid \event \right ],
\end{align*}
with $(\wt{U}^{1}, \dots, \wt{U}^{m})$ being uniformly distributed on $\neighbour_{1} \times \dots \times \neighbour_{m}$. 
Similarly, it follows from the specification of the transition kernel $\hyptranskernel$ that $\pphyp^{m+1}$ has full support on $\hypcompact$ and that its density is bounded below by 
\begin{equation*}
\const_{3}
= 
\min_{ \pos^{\prime} \in  \ball_{\tmax}(\pos),\, \hyp^{\prime} \in \hypcompact,\,  \hyp^{\prime\prime} \in \hypcompact }  Z_{\pos^{\prime},\hyp^{\prime},\hypcompact}^{-1} \hyptransdens\left [ (\pos^{\prime},\hyp^{\prime}), \hyp^{\prime\prime} \right ]   \indicator_{\hypcompact}(\hyp^{\prime\prime}) 
>
0,
\end{equation*}
where $Z_{\pos^{\prime},\hyp^{\prime},\hypcompact}$ is an appropriate normalization constant.
Thus,
\begin{align*}
=\, & \PP \left [ \{  \Psi(\pos, t, \wt{U}^{1}, \dots,\wt{U}^{m}) \in   \tposset \} \cap \{ \pphyp^{m+1} \in \thypset \} \mid \event_{1} \cap \event_{2} \right ]
\\
\geq \,& 
\const_{3} \,  \lebegue ( \hypcompact )\, \PP \left [\{  \Psi(\pos, t, \wt{U}^{1}, \dots, \wt{U}^{m}) \in   \tposset \} \cap \{ \wt{U}^{m+1} \in \thypset \}\right ] \\
= \,& 
\const_{3} \, \lebegue ( \hypcompact )\, \PP \left [\{  \Psi(\pos, t, \wt{U}^{1}, \dots, \wt{U}^{m}) \in \tposset \} \right ] \PP \left [ \{ \wt{U}^{m+1} \in \thypset\}\right ]
\end{align*}
where  $\wt{U}^{m+1}\sim {\rm Uniform}(\hypcompact)$ is a uniform random variable independent of the arrival times $\ta^{k} ~ (k \in \NN)$. 

From the fact that the control sequence was chosen such that all velocity components are flipped at least once, it follows that the Jacobian matrix of the map $(s_{1},\dots,s_{m}) \mapsto \Psi(\pos, t; s_{1}, s_{2},\dots, s_{m} )$ has full rank $\posdim$. Thus, under this map, the pushforward of the uniform law of $(\wt{U}^{1}, \dots, \wt{U}^{m})$ on $\neighbour_{1} \times \dots \times \neighbour_{m}$ is absolutely continuous with respect to the Lebesgue measure on $\posDomain$, and by construction its support contains the point $\ttpos$. Therefore,
 \begin{equation}\label{eq:cont:xi} 
\PP \left [ \Psi(\tpos, \tt, \wt{U}^{1}, \dots, \wt{U}^{m}) \in  \tposset \right ] 
\geq
\const_{3}\,\lebegue( \tposset \cap \ttposset)
\end{equation}
for $\tpos=\pos$ and $\tt=t_{m+1}$, some suitable constant $\const_{3}>0$ and suitable neighbourhood $\ttposset$ of $\ttpos$. This result can be extended to points in a neighbourhood $\posset$ of $\pos$ and an open interval containing $t_{m+1}$ as follows. By viewing $\tpos$ and $\tt$ as parameters of the map $\Psi(\tpos,\tt,\cdot)$, it follows from Lemma 6.3 of \cite{benaim2015qualitative} that there exists a neighbourhood $\tttposset$ of $\ttpos$ and $\varepsilon>0$ such that for $\ttposset=\tttposset$, \cref{eq:cont:xi} holds for all $\tpos \in \posset$ and $\tt \in (t_{m+1} -\varepsilon, t_{m+1}+\varepsilon)$.
Likewise, by virtue of the construction of the constant $\const_{3}$, we have
\begin{equation*}
\PP \left \{ ( \wt{U}^{m+1} \in \thypset ) \right \} 
\geq
\lebegue( \hypcompact)^{-1} \, \lebegue( \thypset \cap \tthypset)
\end{equation*}
for any $\hyp\in \hypcompact$. This completes the proof.  
\end{proof}
Let in the following  $K_{1} \subseteq K_{2} \subseteq \cdots$ be a sequence of a increasing of compact subsets of $\posDomain \times \hypDomain$ such that
$\lim_{n \to \infty} K_{n} = \liminf_{n \to \infty} K_{n}= \posDomain \times \hypDomain$. In accordance with \cite[Section 3]{meyn1993stability}, we define by
\[
\{ \abs{(\ppos,\phyp) } \to \infty \} =  \liminf_{n\to \infty} \liminf_{t\to \infty} \{ (\ppos(t),\phyp(t)) \notin K_{n}\},
\]  
the event that the process escapes to infinity. 
\begin{lemma}\label{lem:non:evan}
The process is {\em non-evanescent}. That is, for any $(\pos,\hyp,\vel)\in\domain$, we have 
\begin{equation*}
\PP_{(\pos,\hyp,\vel)} \left [ \{ \abs{(\ppos,\phyp) } \to \infty \}  \right ]
= 
0.
\end{equation*}
\end{lemma}

\begin{proof}
 
By applying Fatou's lemma twice, we obtain
\begin{align*}
& \quad \PP_{(\pos,\hyp,\vel) \sim\target} \left [ \{ \abs{(\ppos,\phyp) } \to \infty \}  \right ]
\\
& \leq 
\liminf_{n\to \infty} \liminf_{t\to \infty} \PP_{(\pos,\hyp,\vel)}\left [ \{ (\ppos(t),\phyp(t))\notin K_{n}\} \right ]  
% \EE_{(\pos,\hyp,\vel) \sim\target} \left [   \indicator_{\{ (\ppos(t),\phyp(t))\notin K_{n}\}}\right ]
\\
& \leq 
\liminf_{n\to \infty} \{1 - \target(K_{n})\}
= 0,
\end{align*}
where the last equality holds since the target measure $\pi$ is tight. This shows that the process is non-evanescent for $\pi$-almost all starting points $(\pos,\hyp,\vel) \in \domain$. 

We next show non-evanescence for all starting points $(\pos,\hyp,\vel)\in \domain$ by using the fact that the law of $\{\ppos(t),\pvel(t),\phyp(t)\}$ becomes absolutely continuous with respect to $\augtarget$ within finite time. Let $\tau$ be as defined in the first paragraph of \cref{sec:abs:cont} and let $\noneset \subset \domain$ the set of all points in $\domain$ for which the process is non-evanescent. Then,
\begin{align*}
\begin{aligned}
\PP_{(\pos,\hyp,\vel)} \left [ \abs{(\ppos(t),\phyp(t))} \not\to \infty \right ]
& \geq 
\PP_{(\pos,\hyp,\vel)} \left [ \{ \tau < \infty \} \cap \{ \abs{(\ppos,\phyp)} \not\to \infty \} \right ]
\\
& = 
\EE_{(\pos,\hyp,\vel)} \left [ \indicator_{\tau < \infty} \PP_{[\ppos(\tau),\phyp(\tau),\pvel(\tau)]}\{ \abs{(\ppos,\phyp)} \not\to \infty \} \right ]
\\
& \geq
\EE_{(\pos,\hyp,\vel)} \left \{ \indicator_{\tau < \infty}\indicator_{[\ppos(\tau),\phyp(\tau),\pvel(\tau)] \in \noneset } \right \}.
\end{aligned}
\end{align*}
Since by what we have shown above $\domain \setminus \noneset$ is a Lebesgue null set, we have  
\begin{align*}
\EE_{(\pos,\hyp,\vel)} \left \{ \indicator_{\tau < \infty}\indicator_{[\ppos(\tau),\phyp(\tau),\pvel(\tau)] \notin \noneset } \right \} 
& =
\PP_{(\pos,\hyp,\vel)} \left [ \left [ \tau < \infty\right ] \cap \left [  \{ \ppos(\tau),\phyp(\tau),\pvel(\tau)\}  \notin \noneset \right ] \right ]
\\
& = 
\PP_{(\pos,\hyp,\vel)} \left  (\tau < \infty \right).
\end{align*}
Thus, 
\begin{align*}
\PP_{(\pos,\hyp,\vel)} \left [ \abs{(\ppos,\phyp)} \not\to \infty \right ]
\geq
\PP_{(\pos,\hyp,\vel)} \left (\tau < \infty \right ) 
=
1 - \lim_{t\to\infty} \PP(\tau > t)
=
1, 
\end{align*}
since
\begin{align*}
& \quad 
\PP(\tau > t)  \leq \sum_{k=1}^{d+1}\PP(\swtime^{k} > t ) 
\\
& =
\sum_{k=1}^{d+1} \exp \left [ -\int_{0}^{t}\lambda_{i} \{ \ppos(s),\phyp(s),\pvel(s)\} \dd s \right ] \leq  (d+1)e^{-\gammamin t}  \xrightarrow[t \to \infty]{} 0. 
\end{align*}
\end{proof}

With the results of \cref{lem:cont:fix,lem:non:evan} at hand, the proof \cref{thm:ergodic} is identical to the proof of Theorem 5 of \cite{bierkens2019ergodicity}. For the sake of self-contained presentation, we briefly summarize the main steps of that proof, but refer to the original work for details. 

First, Lemmas \ref{lem:cont:fix} and \ref{lem:non:evan} imply the existence of a non-trivial lower semi-continuous sub-stochastic transition kernel $\mathcal{T}$ which bounds the residual kernel
\begin{equation*}
R\{(\pos,\hyp,\vel), \cdot\}
= 
\int_{0}^{\infty}\transkernel_{t}\{(\pos,\hyp,\vel),\cdot\} e^{-t}\dd t
\end{equation*}
from below so that $R\{(\pos,\hyp,\vel), A\} \geq \mathcal{T}\{(\pos,\hyp,\vel), A\}$ for all $(\pos,\hyp,\vel) \in \domain$ and all measurable sets $A  \subset \domain$.
In the language of \cite{tweedie1994topological}, this means that the process is a $T$-process. 

\cref{lem:cont:fix} directly implies that the process  is open set irreducible. That is, for any open set $\mathcal{O}\subset \domain$ and any starting point of the process, the probability that hitting times of the form $\tau_{\mathcal{O}} = \inf\{ t \geq 0, (\ppos(t),\phyp(t),\pvel(t))\in \mathcal{O}\}$ are finite is positive. By Theorem 3.2 of \cite{tweedie1994topological}, the open set irreducibility and the fact that the process is a $T$-process implies that it is $\psi$-irreducible, that is, 
\begin{equation*}
\EE_{(\pos,\hyp,\vel)} \left [\int_{0}^{\infty}\indicator_{A}\{\ppos(t),\pphyp(t),\pvel(t)\} \dd t \right ] >0,
\end{equation*}
for any $(\pos,\hyp,\vel)\in \domain$ and any measurable set $A$ with $\augtarget(A)>0$.

By Theorem 3.2 of \cite{meyn1993stability}, non-evanescence is equivalent to Harris-recurrence in the case of $\psi$-irreducible $T$-processes. Thus, by \cref{lem:non:evan} it follows that the process is Harris-recurrent.

Finally, by Theorem 6.1 of \cite{meyn1993stability}, it is sufficient to show irreducibility of an embedded/skeleton Markov chain 
\begin{equation*}
\state_{k} = \{ \ppos(k\delta), \phyp(k\delta), \pvel(k\delta) \} ~ (k \in \NN)
\end{equation*}
with some $\delta>0$ in order to show  ergodicity of the continuous-time process. The existence of such a Markov chain follows again by \cref{lem:cont:fix} and standard arguments that rely on the observation that any periodicity issues which would prevent the embedded Markov chain to be irreducible can be overcome by the fact that the process can revisit a sufficiently small neighbourhood of any state within a certain non-empty time interval $[t_{0},t_{0}+\varepsilon)$.
\subsection{Proof Theorem \ref{thm:cond:pot}}
To show this, we consider a factorization of $V$ as $V = V_{0} \prod_{i=1}^{\posdim}V_{i}$, where
\begin{equation*}
V_{0}(\pos,\hyp,\vel) 
=
\exp\{ a \pot(\pos,\hyp)\}
\quad \text{and} \quad 
V_{i}(\pos,\hyp,\vel) 
=
\exp [\phi  \{\vel_{i} \partial_{\pos_{i}}\pot(\pos,\hyp) \}] ~~ ( i=1,\dots,\posdim).
\end{equation*}
Let  $s= \sign \{ \vel_{i}\partial_{\pos_{i}}\pot(\pos,\hyp)\}$ and $\widehat{s} = \sign \{ \vel_{i} \partial_{\pos_{i}} \pot(\pos,\widehat{\hyp}) \}$. Since the derivatives $\partial_{\xi_{i}}\bfunc$ are bounded, it follows that there are constants $\const_{i}>0 ~ (i=1,\dots,\posdim)$, such that
\begin{align*}
\begin{aligned}
\frac{V_{i}(\pos,\hyp,\vel)}{V_{i}(\pos,\widehat{\hyp},\vel)} 
&= \frac
{ \left \{ 1+\delta \abs*{\vel_{i}\partial_{\pos_{i}}\pot(\pos,\hyp)} \right \}^{s/2}}
{\left \{ 1+\delta \abs*{\vel_{i}\partial_{\pos_{i}}\pot(\pos,\widehat{\hyp})}\right \}^{\widehat{s}/2}}
=
 \frac
{ \left \{ 1+\delta \abs*{\vel_{i}\partial_{\pos_{i}}\pot_{1}(\pos) + \vel_{i}\partial_{\pos_{i}}\bfunc(\pos,\hyp)}\right \}^{s/2}}
{\left \{ 1+\delta \abs*{ \vel_{i}\partial_{\pos_{i}}\pot_{1}(\pos) + \vel_{i}\partial_{\pos_{i}}\bfunc(\pos,\widehat{\hyp})}\right \}^{\widehat{s}/2}}
\leq \const_{i}
\end{aligned}
\end{align*}
for all $(\pos,\vel) \in \posDomain\times \velDomain$. Thus, in particular 
\begin{equation}\label{eq:lya:eqn:2}
\int_{\hypDomain} \frac{V(\pos,\hyp,\vel)}{V(\pos,\widehat{\hyp},\vel) } \exp\{-\pot(\pos,\hyp)\} \dd \hyp 
\leq
\const \int_{\hypDomain}  \frac{V_{0}(\pos,\hyp,\vel)}{V_{0}(\pos,\widehat{\hyp},\vel) } \exp\{-\pot(\pos,\hyp)\}\dd \hyp
~~ \text{with} ~~ 
\const = \prod_{i=1}^{\posdim} \const_{i}
\end{equation}
for all $\pos\in \posDomain$. By the boundedness of $\bfunc$, it follows that there exist suitable positive constants $\const_{0}, \const_{0}^{\prime}>0$ such that 
\begin{align}\label{eq:lya:eqn:3}
\begin{aligned}
\int_{\hypDomain} \frac{V_{0}(\pos,\hyp,\vel)}{V_{0}(\pos,\widehat{\hyp},\vel) } \exp\{-\pot(\pos,\hyp)\}\dd \hyp 
& \leq
\const_{0}^{\prime} e^{-\pot_{1}(\pos)} \exp\{-a\pot_{2}(\widehat{\hyp})\},
\\
\const_{0} \exp\{-\pot_{1}(\pos)\}
& \leq
\int_{\hypDomain}  \exp\{-\pot(\pos,\hyp)\} \dd \hyp.
\end{aligned}
\end{align}
Since we assume that $U_{1}(\hyp) \to \infty$ as $\abs{\hyp}\to \infty$, inequalities \eqref{eq:lya:eqn:2} and \eqref{eq:lya:eqn:3} imply the validity of \cref{eq:lya:eqn:1}  for sufficiently large $\widehat{\hyp}$.

\subsection{Proof of Theorem \ref{thm:ergodic:2}}

In order to prove \cref{thm:ergodic:2},  we first show the validity of a minorization condition in \cref{lem:minorization} and a Lyapunov condition in \cref{lem:lyapunov}.

\begin{lemma}[Minorization condition]\label{lem:minorization}

Let \cref{as:supp} be satisfied. If $\posset\subset \posDomain$ and $\hypset\subset \hypDomain$ are compact, then there exists $\tt>0$ and a constant $\const>0$ such that 
\begin{equation*}
\forall (\pos,\hyp,\vel) \in C = \posset \times \hypset  \times \velDomain,
\quad
\PP_{(\pos,\hyp,\vel)} \left [ \{\ppos(\tt),\phyp(\tt),\pvel(\tt)\}  \in \cdot \right ]
\geq 
c \, \lebegue( \cdot \cap C).
\end{equation*}

\end{lemma}
\begin{proof}
For any compact set $\posset \subset \posDomain$, we can choose $t>0$ sufficiently large (for example, $t> \posdim \max_{(\pos,\tpos) \in \posDomain \times \posDomain} \abs{\pos-\tpos}_{\infty}$) such that for any pair of points $(\pos,\vel)$ and $(\tpos,\tvel)$ whose position components are contained in $\posset$, there exists an admissible control sequence $\ubf=(\tbf,\ibf)$ with $\tbf = (t_{1},\dots,t_{m+1})$ and $t_{m+1}=t$ connecting $(\pos,\vel)$ and $(\tpos,\tvel)$. By \cref{lem:cont:fix}, for any such pair and any hyperparameter values $\hyp,\thyp$, there exist neighborhoods of $(\pos, \hyp, \vel)$ and $(\tpos,\thyp,\tvel)$ such that \cref{eq:cont:comp} holds for $\tt=t$ and suitable constants. By compactness of $C$, there exists a finite cover of $\posset$ of such neighborhoods, which proves the lemma. 
\end{proof}

\begin{lemma}[Infinitesimal Lyapunov condition]\label{lem:lyapunov}

Let \cref{as:pot} be satisfied. Let $\delta>0$ and $a>0$ be such that $0\leq \gammamax \delta <a  < 1$ with $\gammamax$ as specified in \cref{as:pot}. Further, define $\phi(s) = {\rm sign}(s) \log(1+\delta \abs{s})/2$. Then the function 
\begin{equation*}
V(\pos,\hyp,\vel)  
=
\exp \left [ a \pot(\pos,\hyp) + \sum_{i=1}^{\posdim} \phi \left \{ \vel_{i} \partial_{\pos_{i}}\pot(\pos,\hyp) \right \} \right ]
\end{equation*}
is a Lyapunov function of the GZZ process, that is, $\lim_{x \to \infty}V(x) = \infty$ and there are suitable constants $a>0, b \in \RR$ and a compact set $C \subset \posDomain \times \hypDomain$ such that the Lyapunov condition \eqref{eq:lya:1} is satisfied. 

\end{lemma}
\begin{proof}
We show the validity of the Lyapunov condition
\begin{equation}\label{eq:lya:2}
\Lc V  \leq - r\, V + b\,\indicator_{C},
\end{equation}
with suitable constants $r>0,b\in \RR$, and compact set $C$, separately for $\Lc = \Lczz$ and $\Lc=\Lcg$.
\\

\noindent {\bf (I) $\Lc=\Lczz$:}
For fixed $\hyp\in \hypDomain$, the function $V(\cdot,\hyp,\cdot)$ is identical to the Lyapunov function proposed in Section 3.4 of \cite{bierkens2019ergodicity}, where it is used to show a similar result for the ZZ process. Using the fact that $0 \leq \phi^{\prime}(s) \leq \delta/2$, it is shown in the referenced article that 
\begin{align*}
& \quad 
\Lczz V(\cdot,\hyp,\cdot) 
\\
& \leq
\left \{ -\min(1-a,a-\gammamax \delta)\sum_{i=1}^\posdim \abs*{\partial_{\pos_{i}}\pot(\cdot,\hyp)} + \frac{d}{\delta} + \frac{\delta}{2} \sum_{i,j=1}^\posdim  \abs*{\partial_{\pos_{i}}\partial_{\pos_j}\pot(\cdot,\hyp)} \right \}
V(\cdot,\hyp,\cdot),
\end{align*}
which  under the asymptotic growth condition of \cref{as:pot}  directly implies  the validity of \cref{eq:lya:2} for sufficiently large $C$.
\\

\noindent {\bf (II) $\Lc=\Lcg$:} 
We note that 
\begin{equation*}
\left ( \Lcg V\right)(\pos,\aa,\vel) 
=
V(\pos,\aa,\vel)  \underbrace{\int_{\hypDomain}\left \{ \frac{V(\pos,\aatt,\vel)}{V(\pos,\aa,\vel) } -  1 \right \} \frac{1}{Z_{\pos}} \exp\{-\pot(\pos,\aatt)\} \dd \aatt}_{=: \,C(\pos,\hyp,\vel)}.
\end{equation*}
Thus, in order for the Lyapunov condition to be satisfied, it is sufficient to show that the parameters $a>0$ and $\delta>0$ of the Lyapunov function $V(\pos,\aa,\vel) $ can be chosen such that  there exists $c>0$ and  $r>0$ so that the inequality 
$C(\pos,\hyp,\vel)<-r$
holds for all $\vel \in \{-1,1\}^{\posdim}$ and $(\pos,\aa) \in \posDomain \times \hypDomain$ with $\abs{(\pos,\aa)}>c$. Indeed, this is directly implied by \cref{as:pot}, \eqref{as:pot:it:2}.
\end{proof}

Let $P=\transkernel_{\tt}$ with $\tt$  as specified in \cref{lem:minorization}. By \cref{lem:lyapunov} and a simple Gr\"onwall inequality, it follows that $P$ satisfies a Lyapunov inequality of the form   $ \forall t \,\geq0,~e^{t\Lcgzz}V\leq r V + h\indicator_{C}$ with suitable $r \in (0,1)$ and $h\in \RR$. By Theorem 3.4 of \cite{hairer2011yet}, it follows that the embedded Markov chain associated with $P$ is geometrically ergodic with invariant measure $\augtarget$, that is, 
\begin{equation*}
\forall \, n\in \NN, ~ \forall \, \varphi \in L^{\infty}_{V}(\domain),
\quad
\norm*{P^{n}\varphi - \EE_{x\sim\augtarget}\{\varphi(x)\} }_{L^{\infty}_{V}} 
\leq
\tilde{c} r^{n} \norm*{\varphi - \EE_{x\sim\augtarget}\{\varphi(x)\} }_{L^{\infty}_{V}}
\end{equation*}
with $\tilde{c}>0$. It is well known that geometric ergodicity of the embedded Markov chain together with the validity of an infinitesimal Lyapunov condition implies \cref{thm:ergodic:2} with $\lambda = -\log(r)/\tt $ and sufficiently large constant $\const>0$ (see for example, \citealp[Section 2.4.2.]{lelievre2016partial}).

\section{General construction of Gibbs-PDMP samplers} \label{app:Gibbs-PMDP}
\def\z{\bm z}
\def\trans{\mathcal{P}}
The ZZ-sampler used to update $\pos$ within the GZZ process can be replaced by any other PDMP process that has $\target(\dd \pos \mid \hyp)$ as its invariant measure.  
In what follows, we provide a generic algorithm that generates such a general Gibbs-PDMP sampler. Following the description of a  PDMP in terms of a deterministic flow map, event rate, and transition distribution in \cite{fearnhead2018piecewise}, we assume that the PDMP used for updating $\pos$ is of the form $\z(t) = (\pos(t),\vel(t)) \in \posDomain \times \velDomain$, where $\velDomain \subseteq \RR^{\posdim}$ and is specified by the following $\hyp$-dependent quantities.
%that is     $\z(t) = (\pos(t),\vel(t)) \in \posDomain \times \velDomain$ where $\velDomain \subseteq \RR^{p}$
%The proof of invariance of the corresponding target measure $\wt{\target}$ follows in the same way as in the special case   
\begin{itemize}
\item
\textbf{The deterministic flow map} $\Psi_{\hyp}(\z(s),t) =  \z(t+s)$ that describes the deterministic dynamics of the process between event times.  
\item 
\textbf{The event rate} $\lambda_{\hyp}(\z(t))$ that depends on the current position $\z(t)$ of the process. Event times $T^{k}, k=1,2,\dots$ are generated from the Poisson process $\Pi$ with intensity function $\lambda_{\hyp}(\z(t))$.
\item 
\textbf{The transition distribution at events:} whenever an event occurs at time $\tau$, the state prior to this event is denoted as $\z(\tau-)$, and the state $\z(\tau)$ is sampled from the transition kernel $\trans_{\hyp}$,  that is,  $\z(\tau) \sim \mathcal{P}_{\hyp}( \z(\tau), \cdot)$
\end{itemize}

For any value of $\hyp$, the process specified by the above quantities is assumed to preserve $\target(\dd \pos \mid \hyp)$ (see \citealp{fearnhead2018piecewise} for sufficient conditions). In particular, when considering $\alpha$ as a constant part of the process, this family of processes defines a generator $\mathcal{L}_{{\rm PDMP}}$ that acts on the set of test functions $\test = \C^{\infty}(\posDomain \times \hypDomain \times \velDomain,\RR)$. The generalization of the Gibbs process, that is, the process  associated with the generator $\mathcal{L}_{{\rm Gibbs-PDMP}} = \mathcal{L}_{{\rm PDMP}} + \eta \L_{\rm Gibbs}$, can be verified to  preserve the target measure ${\target}(\dd\pos,\dd \hyp)$ by following the same steps as in the proof of \cref{prop:inv}. It can be simulated as detailed in \cref{alg:GibbsPDMP}.

\begin{algorithm}
\caption{General Gibbs-PDMP algorithm.} 
\label{alg:GibbsPDMP} 
\textbf{Input:} $(\ppos^{0},\phyp^{0}, \pvel^{0}) \in \posDomain \times \hypDomain \times \velDomain$% \jl{order switched compared with main text} and $\swtime^{0} =0$.
\begin{algorithmic}[1] 
\FOR{$k = 1, 2, \dots$}
\STATE 
Draw $\tau' \sim {\rm Exponential}(\swrate)$ and $\wt{\tau}_1, \ldots, \wt{\tau}_\posdim$ such that 
\begin{equation*}
\PP( \wt{\tau} \geq s)
=
\exp \left \{ - \int_0^s  \lambda \left  ( \Psi_{\hyp}(\z(T^{k}),r) \right )  \, \dd r \right \} 
~ (i = 1, \dots, \posdim).
\end{equation*}
\STATE 
Let $\tau^{k} = \min \left \{ \tau', \wt{\tau} \right \}$.  
\STATE 
Set $\swtime^{k+1} = \swtime^{k} +\tau^{k}$. 
\IF{$\tau = \tau'$}
\STATE 
Set  $(\ppos^{k+1},\pvel^{k+1}) =  \Psi_{\hyp}(\z(T^{k}),\tau^{k})$
\STATE 
Draw $\phyp^{k+1} \sim \hyptranskernel\{(\ppos^{k+1},\phyp^{k}),\cdot\}$.
\ELSE
\STATE
Set $\phyp^{k+1} = \phyp^{k}$.
\STATE 
Resample: $(\ppos^{k+1},\pvel^{k+1} ) \sim \trans_{\alpha^{k}}  (\Psi_{\hyp}(\z(T^{k}),\tau^{k}), \cdot ) $,
\ENDIF
\ENDFOR
\end{algorithmic}
\textbf{Output:} Skeleton points $ \{ (\ppos^{k}, \phyp^{k}, \pvel^{k},\swtime^{k})\}_{k \in \NN}$. 
\medskip  
\end{algorithm}

\section{Conditional distributions}  

In the following, $\cdot \mid -$ means conditioned on every variable other than itself. 

\subsection{Random effects model} \label{app:mixed_eff}

We define 
\begin{align*}
\wt{X}_{ij}
& = 
(1, X^\star_{ij1}, \dots, X^\star_{ijK},X_{ij1}, \dots, X_{ij\posdim}) \in \RR^{2+K+\posdim},
\end{align*}
where $X^\star_{ij} = 1$ if observation $i \in j$-th group and zero otherwise. This reduces $\pos \mid -$ to a standard logistic regression setup, which can be sampled from using the ZZ process. In addition, the conditional distributions for the hyperparameters are 
\begin{align*}
\phi \mid -
& \sim
\Ga
\left (a_\phi + \frac{K+1}{2},b_\phi + \frac{m^2}{2} + \frac12 \sum_{j=1}^K \beta_j^2 \right ),
\\
\sigma^2 \mid -
& \sim 
\IG \left (a_\sigma+\frac32, b_\sigma + \frac12 \sum_{l=1}^\posdim \upsilon_i^2 \right ),
\end{align*}
which can be exactly sampled from.

\subsection{Spike-and-slab prior} \label{app:spikeslab_conditionals}

We have the following conditional distributions:
\begin{align*}
p(\nu \mid -) 
& \propto 
p(\upsilon_{1:p} \mid \gamma_{1:d}, \, \nu, \, \tau_{1:d}^2) 
\times p_0(\nu) 
\\
& \propto 
\prod_{i=1}^\posdim \left [ \frac1{\sqrt{\nu} \tau_i \I(\gamma_i=1) + \tau_i \I(\gamma_i=0)} \exp \left \{\frac{-\upsilon_i^2}{2 \left \{ \nu \tau_i^2 \I(\gamma_i=1) + \tau_i^2 \I(\gamma_i=0) \right \}} \right \} \right ] 
\\
& \quad 
\times 
\nu^{-(a_\nu+1)} \exp \left ( \frac{b_\nu}{\nu} \right ) 
\\
& \propto 
\prod_{i:\gamma_i=1} 
\left [ \frac1{\sqrt{\nu} \tau_i} \exp \left ( \frac{-\upsilon_i^2}{2\nu \tau_i^2}  \right ) \right ] 
\times
\nu^{-(a_\nu+1)} \exp \left ( \frac{b_\nu}{\nu} \right ) 
\\
& = 
\IG \left (a_\nu + \frac12 \sum_{i=1}^\posdim \gamma_i, b_\nu + \frac12 \sum_{i=1}^\posdim \frac{\gamma_i \upsilon_i^2}{\tau_i^2} \right );
\\
p(\pi \mid -) 
& = 
\textrm{Beta} \left ( a_\pi + \sum_{i=1}^\posdim \gamma_i, b_\pi + d - \sum_{i=1}^\posdim \gamma_i \right );
\\
\PP(\gamma_i = 1 \mid -) 
& = 
\frac{p( \upsilon_i, \tau_i^2, \pi, \, \nu, \, \gamma_i = 1)}{p(\upsilon_i, \, \tau_i^2, \, \pi, \, \nu)}
= 
\frac{p(\upsilon_i \mid \tau_i^2, \, \nu, \, \gamma_i=1) \times p(\gamma_i=1 \mid \pi) \times p_0(\pi)}{p( \upsilon_i, \, \tau_i^2, \, \pi, \, \nu, \, \gamma_i = 1) + p( \upsilon_i, \, \tau_i^2, \, \pi, \, \nu, \, \gamma_i = 0)}
\\
& = 
\frac{\left(\pi/\sqrt{\nu}\right) \exp \left \{ -\upsilon_i^2 /(2 \nu \tau_i^2) \right \} }{ \left(\pi/\sqrt{\nu}\right) \exp \left \{ -\upsilon_i^2 /(2 \nu \tau_i^2) \right \} + \left(1-\pi \right) \exp \left \{ -\upsilon_i^2 /(2 \tau_i^2) \right \} };
\\
\PP(\gamma_i = 0 \mid -) 
& = 1 -  \PP(\gamma_i = 1 \mid \upsilon_i, \tau_i^2, \pi, \nu) 
\\
& = 
\frac{\left(1-\pi \right) \exp \left \{ -\upsilon_i^2 /(2 \tau_i^2) \right \}}{ \left(\pi/\sqrt{\nu}\right) \exp \left \{ -\upsilon_i^2 /(2 \nu \tau_i^2) \right \} + \left(1-\pi \right) \exp \left \{ -\upsilon_i^2 /(2 \tau_i^2) \right \} }.
\end{align*}

Finally, we consider a consider a MH update step for $\tau_1^2, \dots, \tau_d^2$ by noting that 
\begin{align*}
p_0(\tau_i) 
& \propto 
\left ( 1 + \frac{\tau_i^2}{d_\tau} \right )^{-(d_\tau+1)/2},
\end{align*}
and thus
\begin{align*}
p(\tau_i^2 \mid - )
& \propto 
p(\upsilon_i \mid \tau_i^2, \, \gamma_i, \, \nu) \times p_0(\tau_i^2) 
\\
& = 
p(\upsilon_i \mid \tau_i^2, \, \gamma_i=1, \, \nu) \times p_0(\tau_i^2) 
+ p(\upsilon_i \mid \tau_i^2, \, \gamma_i=0, \, \nu) \times p_0(\tau_i^2)
\\
& \propto 
\frac{\gamma_i}{\tau_i \sqrt{\nu}} \exp \left ( - \frac{\upsilon_i^2}{2 \nu \tau_i^2} \right ) \left ( 1 + \frac{\tau_i^2}{d_\tau} \right )^{-(d_\tau+1)/2}
+
\frac{1-\gamma_i}{\tau_i} \exp \left ( - \frac{\upsilon_i^2}{2 \tau_i^2} \right ) \left ( 1 + \frac{\tau_i^2}{d_\tau} \right )^{-(d_\tau+1)/2}
\\
& = 
\frac{1}{\tau_i} \left ( 1 + \frac{\tau_i^2}{d_\tau} \right )^{-(d_\tau+1)/2} 
\left [  \frac{\gamma_i}{\sqrt{\nu}} \exp \left ( - \frac{\upsilon_i^2}{2 \nu \tau_i^2} \right )
+
(1-\gamma_i) \exp \left ( - \frac{\upsilon_i^2}{2 \tau_i^2} \right ) \right ]
\end{align*}

\end{document}

%% file: macros.tex
\usepackage{mathtools,xparse}
\usepackage{amsfonts}

\DeclarePairedDelimiter{\abs}{\lvert}{\rvert}
\DeclarePairedDelimiter{\norm}{\lVert}{\rVert}

\def\EE{{ \mathbb{E}}}
\def\RR{{\mathbb R}}
\def\NN{{ \mathbb{N}}}
\def\PP{{ \mathbb{P}}}

\def\target{{\pi}}
\def\augtarget{{\widetilde{\target}}}

\def\pot{{\Psi}}

\def\transkernel{{\mathcal{P}}}
\def\hyptransdens{{q}}
\def\hyptranskernel{{\mathcal{Q}}}

\def\pos{{\xi}}
\def\vel{{\theta}}
\def\hyp{{\alpha}}
\def\tpos{\widetilde{\pos}}
\def\tvel{\widetilde{\vel}}
\def\thyp{\widetilde{\hyp}}
\def\ttpos{\widehat{\pos}}
\def\ttvel{\widehat{\vel}}
\def\tthyp{\widehat{\hyp}}
\def\tmax{t_{{\rm max}}}

\def\ball{{B}}

\def\event{\mathcal{E}}

\def\const{c}
\def\wt{\widetilde}

\def\pzeta{{\bm \zeta}}
\def\ppos{{\bm \pos}}
\def\pvel{{\bm \vel}}
\def\phyp{{\bm \hyp}}

\def\pppos{\widehat{\ppos}}
\def\ppvel{\widehat{\pvel}}
\def\pphyp{\widehat{\phyp}}

\def\posset{\mathcal{A}_{\pos}}

\def\hypset{\mathcal{A}_{\hyp}}
\def\tposset{\B_{\pos}}

\def\thypset{\B_{\hyp}}
\def\ttposset{\C_{\pos}}

\def\tthypset{\C_{\hyp}}

\def\tttposset{\widehat{\C}_{\pos}}

\def\lebegue{{\rm Lebesgue}}

\def\ubf{{\bf u}}
\def\tt{{\tilde{t}}}

\def\tbf{{\bf t}}
\def\ibf{{\bf i}}
\def\ta{E}
\def\tnp{\tau}

\def\hypcompact{\widetilde{\Omega}_\hyp}
\def\swtime{T}
\def\trans{{\mathcal{P}}}

\def\pot{{U}}
\def\z{{z}}

\def\Lc{{\mathcal{L}}}

\def\indicator{\bm 1}

\def\gammamin{\underline{\gamma}}
\def\gammamax{\overline{\gamma}}
\def\lambdamax{\overline{\lambda}}
\def\lambdamin{\underline{\lambda}}

\def\bfunc{b}

\def\indicator{{\bm 1}}
\def\neighbour{{\mathcal{U}}}
\def\noneset{\mathcal{N}}
\def\const{c}
\def\state{{x}}
\def\Lcgzz{\Lc_{{\rm GZZ}}}
\def\Lczz{\Lc_{{\rm ZZ}}}
\def\Lcg{\Lc_{{\rm Gibbs}}}

\def\bmax{{\overline{b}}}
\def\bfunc{b}
\def\sign{{\rm sign}}
\def\sparsity{\epsilon}
\def\covdist{\varrho}

\def\velDomain{{\Omega_{\vel}}}
\def\posDomain{{\Omega_{\pos}}}
\def\zetaDomain{{\Omega_{\zeta}}}
\def\hypDomain{{\Omega_{\hyp}}}

\def\dd{{\rm d}}

\def\domain{\Omega}
\def\swrate{\eta}

\def\event{{\mathcal{E}}}

\newcommand{\A}{\mathcal{A}}
\newcommand{\B}{\mathcal{B}}
\newcommand{\C}{\mathcal{C}}

\newcommand{\G}{\mathcal{G}}

\renewcommand{\L}{\mathcal{L}}

\newcommand{\N}{\mathrm{Normal}}

\renewcommand{\NN}{\mathbb{N}}

\renewcommand{\RR}{\mathbb{R}}

\newcommand{\IG}{\mathrm{IG}}

\newcommand{\iid}{\stackrel{\rm iid}{\sim}}  
\newcommand{\ind}{\stackrel{\rm ind}{\sim}}

 \DeclareMathOperator*{\argmin}{\arg\!\min}

\def\nobs{n}
\def\hypdim{r}
\def\posdim{p}
\def\zetadim{d}
 
\def\1{{ \mathbf{1}}}
 
\def\Ga{\mathrm{Ga}}

\def\scaleinvchi2{\mathrm{Scale}\text{-}\mathrm{inv}\text{-}\chi^2}

\def\test{\mathscr{S}}
\def\I{\mathbb{I}}

\def\b0{\mathbf{0}}

\def\true{\mathrm{true}}